\documentclass[11pt,a4paper]{amsart}
\usepackage[margin=1in]{geometry}
\usepackage{amsmath, amssymb, latexsym, cmll, stmaryrd, eqnarray}
\usepackage{scalerel}
\usepackage{mathtools}
\usepackage{xspace}
\usepackage{xcolor}
\usepackage{ifthen}
\usepackage{hyperref}
\usepackage[capitalise]{cleveref}
\usepackage{lineno}
\usepackage{amsaddr}
\usepackage{wrapfig}
\usepackage[normalem]{ulem}

\usepackage{xcolor}

\newtheorem{lm}{Lemma}[section]
\newtheorem{fact}[lm]{Fact}
\newtheorem{thm}[lm]{Theorem}

\newtheorem{prob}{Problem}

\newcounter{senumi}[section]
\newcounter{senumip}[section]
\newcounter{temp}[section]

\def\thesenumi{\thesection.\arabic{senumip}}
\def\p@senumip\thesenumip{\thesenumi}
\newenvironment{senumerate}%
    {\begin{list}%
        {\hspace{-2em}(\thesenumi)}%
        {\usecounter{senumip}}
        \setcounter{senumip}{\value{temp}}
    }%
    {\setcounter{temp}{\value{senumip}}
     \end{list}}
\newcounter{penumi}[section]

\newcounter{ptemp}[section]

   {\begin{enumerate}%
        \setcounter{penumi}{\value{ptemp}}%
        \setcounter{enumi}{\value{ptemp}}%
    }%
    {\setcounter{ptemp}{\value{enumi}}
     \end{enumerate}}
\newcounter{ppenumi}[section]

\newcounter{pptemp}[section]

\def\theppenumi{\theptemp.\arabic{ppenumi}}
    {\begin{list}%
        {(\theppenumi)}%
        {\usecounter{ppenumi}\setlength{\rightmargin}{\leftmargin}}
        \setcounter{ppenumi}{\value{pptemp}}
    }%
    {\setcounter{pptemp}{\value{ppenumi}}
     \end{list}}

\newcommand{\polsat}[1]{\textsc{PolSat}\left( {\m #1} \right)}

\newcommand{\npc}{\textsf{NP}-complete\xspace}
\newcommand{\conpc}{\textsf{co-NP}-complete\xspace}
\newcommand{\ptime}{\textsf{P}\xspace}
\newcommand{\rptime}{\textsf{RP}\xspace}
\newcommand{\prptime}{\textsf{(R)P}\xspace}

\newcommand{\m}[1]{{\uppercase {\bf{#1}}}}

\newcommand{\set}[1]{{\left\{ {#1} \right\} }}
\newcommand{\ci}{\subseteq}

\newcommand{\card}[1]{\left| #1 \right|}

\newcommand{\intv}[2]{I\left[#1,#2\right]}
\newcommand{\vpair}[2]{{{#1}\choose{#2}}}

\renewcommand{\leq}{\leqslant}
\renewcommand{\geq}{\geqslant}

\renewcommand{\mapsto}{\longmapsto}

\newcommand{\join}{\vee}
\newcommand{\meet}{\wedge}

\newcommand{\con}[1]{{\sf Con\:\m{#1}}}

\newcommand{\pol}[1]{{\rm Pol\:\m #1}}

\newcommand{\po}[1]{{\mathbf {#1}}}

\newcommand{\tn} [1]{{\bf {#1}}}
\newcommand{\typ}{{\rm typ}}
\newcommand{\typset}[1]{\typ\set{#1}}

\renewcommand{\o}[1]{\overline {#1}}
\newcounter{ttable}

\newcommand{\comm}[2]{\left[ #1 , #2 \right]}

\newcommand{\map}{\longrightarrow}

\newcommand{\congruent}[1]{\stackrel{#1}{\equiv}}

\newcommand{\h}[1]{\widehat{#1}}

\newcommand{\fj}{\varphi}

\usepackage{todonotes}

\newcommand{\mute}[1]{}

\newcommand{\gProblem}[2]{\ensuremath{\operatorname{\textup{\textsc{{#2}}}}
		\ifthenelse{\equal{#1}{}}{}{\!\left( {#1} \right)}}}

\renewcommand{\polsat}[1]{\gProblem{#1}{PolSat}}
\newcommand{\poleqv}[1]{\gProblem{#1}{PolEqv}}

\newcommand{\ceqv}[1]{\gProblem{#1}{CEqv}}
\newcommand{\csat}[1]{\gProblem{#1}{CSat}}

\newcommand{\progpolsat}[1]{\gProblem{#1}{ProgSat}}
\newcommand{\progcsat}[1]{\gProblem{#1}{ProgCSat}}

\newcommand{\cm}{congruence modular }

\newcommand{\pupi}{PUPI }

\newcommand{\sr}[1]{{\mathsf{sr}\left(#1\right)}}

\newcommand{\charr}{\mathsf{char}}
\newcommand{\charrset}[1]{\charr\set{#1}}

\newcommand{\setm}{-}

\newcommand{\cc}{c}

\newcommand{\ee}{e}

\newcommand{\aai}{a_i}

\newcommand{\cci}{c_i}
\newcommand{\ddi}{d_i}

\newcommand{\alphai}{\alpha_i}

\newcommand{\betai}{\beta_i}

\newcommand{\alphami}{\alpha^-_i}

\newcommand{\gammai}{\gamma_i}
\newcommand{\gammao}{\gamma_{1-i}}
\newcommand{\gammaz}{\gamma_0}
\newcommand{\gammaj}{\gamma_1}

\newcommand{\vi}{V_i}

\newcommand{\vz}{V_0}
\newcommand{\vj}{V_1}

\newcommand{\fji}{\fj_i}
\newcommand{\fjpi}{\fj^+_i}
\newcommand{\psii}{\psi_i}

\newcommand{\alpham}{\alpha^-}

\newcommand{\z}{\mathbb{Z}}
\newcommand{\N}{\mathbb{N}}

\DeclareMathOperator*{\amper}{\scalerel*{\&}{\sum}}

\newcommand{\matr}[2]{\mbox{End}{\left(\z_{#1}^{#2}\right)}}

\newcommand{\progg}[4]{\left(#1,#2,#3,#4\right)}
\newcommand{\prog}[2]{\left(#1\right)\!\left[#2\right]}
\newcommand{\progb}[3]{\left(#1\right)\!\left[#2,#3\right]}
\newcommand{\fcirc}[1]{#1^\circ}
\newcommand{\beq}{\mbox{\tt ==}}

\renewcommand{\b}{\textsf{b}}
\newcommand{\true}{1}
\newcommand{\false}{0}
\newcommand{\bool}{\set{\false,\true}}

\newcommand{\sacik}{\po {sat}}

\newcommand{\sumpk}[2]{\Sigma_{(#1,#2)}}
\newcommand{\aexp}{a}           

\newcommand{\zero}{e}

\newcommand{\ccc}{c}    
\newcommand{\s}{s}      

\newcommand{\csize}{\lambda}   

\newcommand{\ccand}{\mathsf{AND}}
\newcommand{\ccmod}{\mathsf{MOD}}
\newcommand{\ccor}{\mathsf{OR}}

\newcommand{\sdiv}[1]{\Delta\left(#1\right)}    
\newcommand{\ppdiv}{\delta}
\newcommand{\pdiv}[1]{\ppdiv\left(#1\right)}    
\newcommand{\ar}[1]{\mu\left({#1}\right)}       
\newcommand{\maxar}[1]{\mu\left({\m #1}\right)} 

\newcommand{\cdhh}{CDH\xspace}

\newcommand{\ethh}{ETH\xspace}
\newcommand{\rethh}{rETH\xspace}
\newcommand{\prethh}{(r)ETH\xspace}

\begin{document}

\title{Nonuniform Deterministic Finite Automata \\over finite algebraic structures}

\author{Paweł M. Idziak}
\address{Department of Theoretical Computer
Science,\\ Jagiellonian University,\\
Kraków, Poland}
\email{pawel.idziak@uj.edu.p}

\author{Piotr Kawałek}
\address{Institute of Discrete Mathematics and Geometry,\\ TU Wien, Austria \\[10pt] Department of Theoretical Computer
Science,\\ Jagiellonian University,\\
Kraków, Poland}
\email{piotr.kawalek@tuwien.ac.at}

\thanks{Piotr Kawałek: This research was funded in whole or in part by National Science Centre, Poland \#2021/41/N/ST6/03907. For the purpose of Open Access, the author has applied a CC-BY public copyright licence to
 any Author Accepted Manuscript (AAM) version arising from this submission.
Funded by the European Union (ERC, POCOCOP, 101071674). Views
 and opinions expressed are however those of the author(s) only and do not necessarily reflect those
 of the European Union or the European Research Council Executive Agency. Neither the European
 Union nor the granting authority can be held responsible for them.}

\author{Jacek Krzaczkowski}

\address{Department of Computer Science,\\
Maria Curie-Skłodowska University,\\
Lublin, Poland}
\email{krzacz@umcs.pl}
\thanks{Jacek Krzaczkowski: This research was funded in whole or in part by National Science Centre, Poland \#2022/45/B/ST6/02229. For the purpose of Open Access, the author has applied a CC-BY public copyright licence to any Author Accepted Manuscript (AAM) version arising from this submission.}

\begin{abstract}
Nonuniform deterministic finite automata (NUDFA) over monoids were invented by Barrington in \cite{Barrington85} to study boundaries of nonuniform constant-memory computation. Later, results on these automata helped to indentify interesting classes of groups for which  equation satisfiability problem (\polsat{}) is solvable in (probabilistic) polynomial-time \cite{GoldmannR02, IdziakKKW22-icalp}. Based on these results, we present a full characterization of groups, for which the identity checking problem (called \poleqv{}) has a probabilistic polynomial-time algorithm. We also go beyond groups, and propose how to generalise the notion of NUDFA to arbitrary finite algebraic structures. We study satisfiability of these automata in this more general setting. As a consequence, we present full description of finite algebras from congruence modular varieties for which testing circuit equivalence $\ceqv{}$ can be solved by a probabilistic polynomial-time procedure. In our proofs we use two computational complexity assumptions: randomized Expotential Time Hypothesis and Constant Degree Hypothesis.
\end{abstract}

\maketitle

\newpage

\section{Introduction}
There are many interactions between mathematics and (theoretical) computer science. 
Many branches of these sciences influence each other, sometimes in quite surprising ways.
Relatively recent example of such an influence is so-called algebraic approach to Constraint Satisfaction Problem (CSP) which led to complete classification of computational complexity of CSP \cite{Bulatov-dichotomy,Zhuk-dichotomy}.
It is really impressive how in this case (universal) algebra, combinatorics, logic, computational complexity and algorithmic work together to give new results in each of these fields.

Another example of synergy between different fields of mathematics and theoretical computer science can be observed on the borderline of circuit complexity, automata theory and (universal) algebra.
The most significant example here is the role of monoids played in automata theory and formal languages.

Usually deterministic finite automaton (DFA) is determined by 
an alphabet $\Sigma$
acting over a set $Q$ of states
by a function $\delta : \Sigma \times Q \ni (\sigma,q) \mapsto \sigma\cdot q \in Q$.
This action can be extended (in an obvious way) to the action of the free monoid $\Sigma^*$. 
To decide if a word $w\in\Sigma^*$ is accepted by a particular DFA
we need to endow it with a starting state $q_0$ and a set $F\ci Q$ of accepting states.  
Then $w$ gets accepted if $w\cdot q_0\in F$.
For our purposes we prefer, first to treat the set $Q^Q$ of functions 
as the monoid with $f\cdot g = g\circ f$, 
and then to treat the action $\delta$ as a function $a : \Sigma \map Q^Q$
given by $a(\sigma)(q)=\sigma\cdot q$.
Now, the word $\sigma^1\ldots\sigma^n$ gets accepted 
if $a(\sigma^1)\cdot\ldots\cdot a(\sigma^n)\in S$, 
where $S$ consists of transitions determined by the words $w\in\Sigma^n$ 
satisfying $w\cdot q_0\in F$.

Before restating Barrington's definition of Non-uniform Deterministic Finite Automata\break (NUDFA) over monoids \cite{Barrington86} we note that $a(\sigma^1)\cdot\ldots\cdot a(\sigma^n)$
is nothing else but\break $\po t_n(a(\sigma^1),\ldots,a(\sigma^n))$, 
where $\po t_n(x_1,\ldots,x_n)=x_1\cdot\ldots\cdot x_n$. 
In NUDFA over the monoid $\m M$ to accept the word from $\Sigma^n$ 
we are going to relax the term $x_1\cdot\ldots\cdot x_n$
to an arbitrary semigroup term $\po t(x_1,\ldots,x_k)$ 
and replace one action $a:\Sigma\map Q^Q$ by a bunch of functions $a^x:\Sigma\map M$,
one for each variable of $\po t$. 
Now a $\po t$-program (with inputs from $\Sigma^n$ represented by $n$-variable word $b_1\ldots b_n$) 
consists of:
\begin{itemize}
  \item a set of $k$-instructions, one for each variable $x$ of $\po t$, 
        of the form $\iota(x)=(b^x,a^x)$, 
        where $b^x$ is one of the variables $b_i$,
  \item and a set $S\ci M$ of accepting values.
\end{itemize}
Finally, a NUDFA over $\m M$  is a sequence (possibly even nonrecursive) of programs
$(\po t_n, n, \iota_n, S_n)_{n\in\N}$ 
with $\po t_n=(x_1,\ldots,x_{k_n})$ being some terms of $\m M$, 
$S_n \ci M$ and $\iota_n$ being the instructions for the variables of $\po t_n$. 
A word $b^1\ldots b^n\in\Sigma^n$ gets accepted by such a NUDFA if\break
$\po t_n(a^{x_1}(b^{x_1}),\ldots,a^{x_{k_n}}(b^{x_{k_n}}))\in S_n$. 
Originally Barrington considered mainly the Boolean case where $\Sigma=\set{0,1}$ 
to study computational boundaries of non-uniform constant-memory computation. 
Such automata can be used to compute the Boolean functions of the form $\set{0,1}^n\map\set{0,1}$. 
They also give an interesting algebraic insight into the internal structure of the class $NC^1$ \cite{BarringtonST90}. Note that in the Boolean case $\Sigma = \{0,1\}$ the function $a^x$ is given by the pair of values $a^x(0), a^x(1)$. In such the case,  for simplicity we write $\iota(x) = (\b^x,a^x(0),a^x(1))$.

A natural question that arises here is 
whether the language accepted by a particular NUDFA over $\m M$ is nonempty.
This problem reduces to:
\begin{itemize}
  \item[]$\progpolsat{\m M}$: \quad
        Decide if a given program over $\m M$ accepts at least one word.
\end{itemize}
The problem $\progpolsat{}$ proved itself to be extremely useful in studying groups (and monoids)
for which determining if an equation has a solution (\polsat{}) is in $\ptime$.
In fact the proof of Goldmann and Russell in \cite{GoldmannR02}
that each nilpotent group $\m G$ has tractable $\polsat{\m G}$
modifies a polynomial time algorithm for $\progpolsat{\m G}$.
A study of connections between \polsat{} and \progpolsat{} for finite monoids is given in \cite{BarringtonMMTT00}.
Recently a complete classification of finite groups $\m G$ with $\progpolsat{\m G} \in \rptime$, together with its consequences for \polsat{} has been provided by Idziak, Kawałek and Krzaczkowski in \cite{IdziakKKW22-icalp}.

Now the results of \cite{IdziakKKW22-icalp} allows us to complete the long term extensive investigations \cite{BurrisL04, HorvathS11, FoldvariH19, Weiss20, IdziakKKW22TOCS}
on equivalence problem \poleqv{} of polynomials (i.e. terms with some variables already evaluated) over finite groups.

The lower bound in our characterization relies
on the randomized version of the Exponential Time Hypothesis (\rethh),
while the upper bounds explore the so-called Constant Degree Hypothesis (\cdhh).
This hypothesis, introduced in \cite{BarringtonST90}, can be rephrased to state
that for a fixed integer $d$, a prime $p$ and an integer $m$ which is not just a power of $p$,
any 3-level circuits of the form $\ccand_d\circ\ccmod_m\circ\ccmod_p$
require exponential size to compute $\ccand_n$ of arbitrary large arity $n$ (in which $\ccand_d$ gates are in the input layer, $\ccmod_m$ gates are in the middle layer, and there is one $\ccmod_p$ gate in the output layer). 
The best known lower bound, due to Chattopadhyay et al. \cite{chat-lowerbounds}, 
for the size of $\ccand_d\circ\ccmod_m\circ\ccmod_p$ computing $\ccand_n$
is only superlinear.
Much earlier \cdhh has been considered in many different contexts.
Already in \cite{BarringtonST90} the case $d=1$, i.e. no $\ccand_d$ layer, has been confirmed. 
Also restrictions put either on the number of $\ccand_d$ used locally, 
or on the local structure of $\ccand_d\circ\ccmod_m$ fragments, 
allowed Grolmusz and Tardos \cite{GrolmuszT00, Grolmusz01} to confirm \cdhh. 
Very recently Kawałek and Wei\ss{}  \cite{KW-symmetric-gates} confirmed \cdhh for symmetric circuits. 

In this paper, under these two complexity hypothesis (i.e. \rethh and \cdhh), the characterization of finite groups with tractable \progpolsat{} from \cite{IdziakKKW22-icalp} is applied  
to get unexpectedly different characterization for tractable \poleqv{}. 

\begin{thm}
\label{eqv-groups}
Let $\m G$ be a finite group.
Assuming \rethh and \cdhh
the problem \poleqv{\m G} is in \rptime if and only if $\m G$ is solvable and has a  nilpotent normal subgroup $\m H$ with the quotient $\m G/\m H$ being also nilpotent.
\end{thm}

A surprising part of these investigations is that such characterization can not only be done,
but that it can be stated, in terms of algebraic structure of the groups.
In fact this reveals another connection between (universal) algebra
and circuit complexity theory. 

The goal of this paper is to leave the group realm and generalize \cref{eqv-groups} 
to a much broader setting of algebras. As we will see shortly this generalization is two-fold.
First we generalize the concept of NUDFas to cover automata working over arbitrary finite algebraic structure $\m A$.
This requires to define a program over $\m A$ and can be done by simply replacing the monoid $\m M$ by the algebra $\m A$ and assume that this time $\po t$ is a term of the algebra $\m A$.

\begin{wrapfigure}{l}{0.39\textwidth}
\begin{center}
\includegraphics[width=0.93\linewidth]{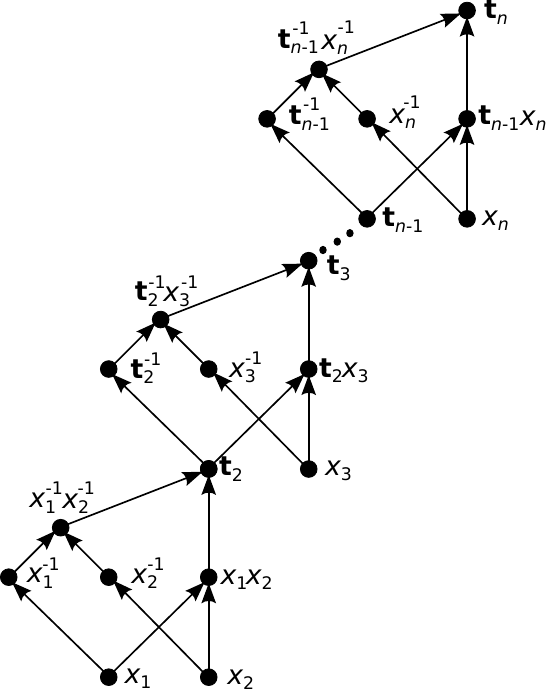}
{\small Compressing the size of $\po t_n$.}
\\{\small \copyright Idziak, Krzaczkowski \cite{IdziakK22}}
\label{fig:comm-circuit}
\end{center}

\end{wrapfigure}

Second, in general algebraic context it is not clear which operations are to be chosen to be the basic ones.
And this choice may be extremely important from computational point of view.
Recall after \cite{IdziakK22}, that adding the binary commutator operation 
$[x,y]=x^{-1}y^{-1}xy$ to the language of a group may exponentially shorten the size of the input.
Indeed the term $\po t_n(x_1,\ldots,x_n) =[\ldots [[x_1,x_2],x_3] \ldots, x_n]$ has linear size if the commutator operation is allowed, while after writing this term in the pure group language (of multiplication and the inverse) we see that the size $\card{\po t_n}$ of $\po t_n$ is 
$2\card{\po t_{n-1}}+2$, as 
$\po t_n(x_1,\ldots,x_n)=
\po t_{n-1}(x_1,\ldots,x_{n-1})^{-1}\cdot x_n^{-1}\cdot\po t_{n-1}(x_1,\ldots,x_{n-1})\cdot x_n$, 
so that $\card{\po t_n}$ is exponential on $n$.
Actually for the alternating group $\m A_4$ it is shown in \cite{HorvathS12} that 
$\polsat{\m A_4}$ is in $\ptime$, while after endowing $\m A_4$ with the commutator operation 
(and therefore shortening the size of inputs) the problem becomes \npc.
A solution to this phenomena has been proposed by Idziak and Krzaczkowski in \cite{IdziakK22} by presenting a term by the algebraic circuit that computes it. 
For example $\po t_n(x_1,\ldots,x_n)$ can be computed by the circuit of size $6n-5$ as shown by Figure \ref{fig:comm-circuit}.

\medskip\noindent
Representing polynomials of an algebra with algebraic circuits leads to a modified version of \progpolsat{} which we explore in this paper.
\begin{itemize}
  \item[]$\progcsat{\m A}$: \quad
        Decide if a given program $(\po t, n, \iota, S)$  over $\m A$ accepts at least one word, where $\po t$ is
        given by a circuit over $\m A$. 
\end{itemize}

Measuring the size of the input, i.e. the expression of the form $\po p=\po q$ by the lengths of the polynomials $\po p$ and $\po q$ or by the sizes of their algebraic circuits used to compute $\po p$ and $\po q$ give rise to either \polsat{} or \csat{} in the satisfiability setting 
or to \poleqv{} or \ceqv{} in the equivalence setting.
Note here that \cite{IdziakK22} argues that the complexity of $\csat{\m A}$ and $\ceqv{\m A}$ is independent of which term operations of the algebra $\m A$ are chosen to be the basic ones.
This independence gives a hope for a characterization of algebras $\m A$ with tractable 
$\csat{\m A}$ or $\ceqv{\m A}$ in terms of algebraic structure of $\m A$. 
Actually already a series of papers 
\cite{IdziakKK20,Weiss20, IdziakKK22STACS,IdziakKKW22-icalp,IdziakK22,Kompatscher21} 
enforces a bunch of such necessary algebraic conditions for an algebra to have \csat{} or \ceqv{} tractable.
Not surprisingly solvability and nilpotency are among these conditions.

To define these two notions of solvability and nilpotency outside the group realm we need a notion of a commutator. However we need to work with a commutator $\comm{\alpha}{\beta}$ of two congruences $\alpha,\beta$ (that in the group setting correspond to normal subgroups) instead of a commutator of elements of an algebra. 
For more details on the definition and the properties of commutator, we refer to the book \cite{fm} and Section \ref{sec-notions-1}.
Here we only note that this concepts of commutator of congruences works smoothly only in some restricted setting of the so called congruence modular varieties, i.e. equationally definable classes of algebras with modular congruence lattices.
Fortunately this setting includes groups, rings, quasigroups, loops, Boolean algebras, Heyting algebras, lattices and almost all algebras related to logic.
In groups, rings or Boolean/Heyting algebras the congruences are determined by normal subgroups, ideals or filters respectively.
Obviously this new concept of commutator of normal subgroups coincides with the old one. 
The commutator of two ideals $I,J$ of a commutative ring 
is simply their algebraic product $I\cdot J$, 
while the commutator of filters in a Boolean/Heyting algebra is their intersection. 
Now we can say that a congruence $\alpha$ is abelian, nilpotent or solvable if  $\comm{\alpha}{\alpha}=0_{\m A}$, 
$\comm{\ldots\comm{\comm{\alpha}{\alpha}}{\alpha}}{\ldots\alpha}=0_{\m A}$
or
$\comm{\comm{\comm{\alpha}{\alpha}}{\comm{\alpha}{\alpha}}}
{\ldots\comm{\comm{\alpha}{\alpha}}{\comm{\alpha}{\alpha}}}=0_{\m A}$ respectively (for some number of nested commutators). 
Here, $0_{\m A}$ is the identity relation/congrueence of $\m A$.
The algebra $\m A$ itself is said to be abelian, nilpotent or solvable if the total congruence $1_{\m A}$ collapsing everything is abelian, nilpotent or solvable, respectively.

Note that, for nilpotent groups, boolean programs (of NUDFAs) compute $\ccand$ functions only of bounded arity, 
i.e. for each nilpotent group $\m G$ there is a constant $k$ such that $\ccand_k$ is computable by no program over $\m G$ \cite{BarringtonST90}. 
This nonexpressibility phenomena does not transfer to nilpotent algebras in general congruence modular context. The most natural example here is the algebra
$(\z_6; +,\%2)$, i.e. the group $(\z_6; +)$ endowed with the unary operation $\%2$ computing the parity.
In this algebra all the circuits of the form $\ccmod_2\circ\ccmod_3$ can be modelled so that $\ccand_n$ can be expressed for all $n$ (however by exponential size of the circuits).
This action of the prime $2$ acting over prime $3$ cannot occur in nilpotent groups.
Indeed, due to the Sylow theorem, each finite nilpotent group is a product of $p$-groups.
This decomposition prevents interaction between different primes, as they occur on different stalks/coordinates. And the lack of such interactions is crucial in bounding the arity of expressible $\ccand_n$'s. 
Also in our considerations the finite nilpotent algebras that decompose into a product of algebras of prime power order occurs naturally. 
They are known as supernilpotent  algebras  \cite{bulatov-kom, aichmud-2010}.
We will return to this concept of supernilpotent algebras and its relativization to supernilpotent congruences in Section \ref{sec-notions-1}. 

Now we are ready to state the other two main results of the paper.

\begin{thm}
\label{thm:cm-progcsat}
Let $\m A$ be a finite algebra from a congruence modular variety. 
Assuming \rethh and \cdhh
the problem $\progcsat{\m A}$ is in \rptime if and only if $\m A$ is nilpotent 
and has a supernilpotent congruence $\sigma$ with supernilpotent quotient $\m A/\sigma$ 
and such that all cosets of $\sigma$ have sizes that are powers of the same prime number $p$.
\end{thm}

\begin{thm}
\label{thm:cm-ceqv}
Let $\m A$ be a finite algebra from a congruence modular variety. 
Assuming \rethh and \cdhh
the problem $\ceqv{\m A}$ is in \rptime if and only if $\m A$ is nilpotent 
and has a supernilpotent congruence $\sigma$ with supernilpotent quotient $\m A/\sigma$. 
\end{thm}

Results from \cite{IdziakKKW22-icalp} that we use in the proof of \cref{eqv-groups}
heavily rely on a method of representing terms/polynomials of finite solvable group $\m G$
by bounded-depth circuits that use only modular gates.
Besides the obvious requirement that the circuit representing a group-polynomial $\po p$ has to compute the very same function as $\po p$ does, we also want to control the size of the circuit to be polynomial in terms of the size (length) of $\po p$.

A very similar approach can be found in \cite{Kompatscher19CC}, where M.\! Kompatscher considers generalizations of finite nilpotent groups, i.e.\! nilpotent algebras from the congruence modular varieties. He provides a method to rewrite circuits over such algebras to constant-depth circuits which, again, use only modulo-counting gates. These modular circuits, which appear in both of the mentioned cases, are known as CC-crcuits. However, since \cite{Kompatscher19CC} does not use the notion of a program/NUDFA, the author formulates his results for circuits over $\m A$ representing only some specific functions. For those functions, it is natural how to interpret Boolean values $0/1$ in the non-Boolean algebra $\m A$. In our paper, the notion of a program/NUDFa provides us with a formal framework which helps to relate the expressive power of algebraic structures to some standard circuit complexity classes. Here, we present a very precise characterization of functions computable by programs over algebras corresponding to polynomial-time cases of \progcsat{}.

\begin{thm}
\label{thm:2supernil-circuit-early}
Let $\m A$ be a finite nilpotent algebra from a congruence modular variety with  supernilpotent congruence $\sigma$ of $\m A$
such that cosets of $\alpha$ are of prime power size $p^{k}$ and  $\m A/\sigma$ is supernilpotent.
Then the function computable by a boolean program of size $\ell$ over the algebra $\m A$
can be also computed by an
$\ccand_d\circ\ccmod_{m}\circ\ccmod_p$-circuits of size $O(\ell^c)$
with $d,m,c$ being natural numbers depending only on $\m A$, and $m$ being relatively prime to $p$.
\end{thm}

This theorem not only is an interesting result on its own, but also is crucial in proving \cref{thm:cm-progcsat} and \ref{thm:cm-ceqv}.

\section{Algebraic preliminaries}\label{sec-notions-1}
An algebra is a set called universe together with a finite set of operations acting on it called basic operations of the algebra. We usually use boldface letter to denote the algebra and the very same latter, but with a regular font, to denote its universe. $\pol{A}$ is polynomial clone of $\m A$, that is the set of all polynomial operations of $\m A$. An algebra $\m A$ is polynomially equivalent to an algebra $\m B$ if it is isomorphic
to an algebra which has the same set of polynomial operations as $\m B$.  An induced algebra $\m A|_S$ is a set $S$  with all polynomial operations of $\m A$ closed on $S$ (or in other word polynomial operations which for arguments from $S$ have value in $S$).  Idempotent  function is a function $f$ such that $f(f(x))=f(x)$.

In proves of intractability of \progcsat{} and \ceqv{} for algebras from congruence modular varieties the crucial role is played by Tame Congruence Theory (see \cite{hm} for details). This is a deep algebraic tool describing local behavior of finite algebras. TCT shows that locally  finite algebras behave in one of following five ways:
\begin{enumerate}
\item[{\tn 1}.]  a finite set with a group action on it,
\item[{\tn 2}.]  a finite vector space over a finite field,
\item[{\tn 3}.]  a two-element Boolean algebra,
\item[{\tn 4}.]  a two-element lattice,
\item[{\tn 5}.]  a two-element semilattice.
\end{enumerate}
By $\typset{\m A}$, let us denote a subset of $\{\tn 1,.., \tn 5\}$ which describes the local behaviours we can find in an algebra $\m A$. Note that in a case of algebra $\m A$ from a congruence modular variety only three types can appear in $\typset{\m A}$, that is types {\tn 2}, {\tn 3} and {\tn 4}. In case of types {\tn 3} and {\tn 4} there has to be two-element set $U$ (let's call element of $U$ as $0$ nad $1$) such that
\begin{itemize}
    \item there is a polynomial $\po e$ of $\m A$ fulfilling $\po e(A)=U$,
    \item there are polynomials of $\m A$ which behaves on $U$ like $\meet$ and $\join$,
    \item in case of type {\tn 3} there is also unary polynomial of $\m A$ which is a negation on $U$.
\end{itemize}

If we can find type {\tn 3} or {\tn 4} in $\typset{\m A}$, the complexity of both \progcsat{} and \ceqv{} is relatively easy to determine, as we shall see in forthcoming chapters. For this reason the most of the volume of the paper is devoted to algebras with $\typset{\m A} = \{\tn 2\}$. In the congruence modular variety those are precisely the solvable algebras \cite{fm, hm}. In fact, some of the earlier papers already dealt with algebras that are solvable but not nilpotent, so we will be mostly concerned with the notion of nilpotency and its properties.

Every solvable (so in particular - nilpotent) algebra in the congruence modular variety is Malcsev, i.e.\! it possesses a polynomial operation $\po d$ satisfying the following identity:
$\po d(y,x,x)=\po d(x,x,y)=y$. A standard example of Malcev algebras are groups with a Malcev term of the form $x \cdot y^{-1} \cdot z$. Unlike for groups, nilpotent algebras do not necessarly decompose into a direct product  of algebras of prime power order. However, the technique contained in this paper splits an algebra into slices on which such nice decomposition can be observed. 

To define this slicing properly, we need to consider congruences. Recall that congruence $\sigma$ of an algebra $\m A$ is an equivalence relation which is preserved by the operations of $\m A$. Such relations are naturally associatted with surjective homomorphisms from $\m A$ to $\m A/\sigma$ mapping $x$ to $[x]_{\sigma}$ (equivalence class of $x$ in $\sigma$), so congruencess are essentially generalization of normal subgroups of a group. Similarly as normal subgroups, congruences of an algebra form a lattice. From now on, we write $\con{\m A}$ for the set of all congruences of $\m A$. Every element of a finite lattice can be written as a meet (join) of meet-irreducible (join-irreducible) elements, i.e.\! elements which cannot be writen as a meet (join) of any other two elements of the lattice. These special elements, generating $\con A$, will play a significant role in our analysis.

For $\alpha, \beta \in \con{\m A}$, by $\intv{\alpha}{\beta}$ we mean a set of congruencess $\gamma$ such that $\alpha \leq \gamma \leq \beta$. In case when $\intv{\alpha}{\beta} = \{\alpha, \beta\}$, i.e.\! there are no congruencess between $\alpha$ and $\beta$, we call $\beta$ a cover of $\alpha$,  we call $\alpha$ a subcover of $\beta$, and we call $\alpha, \beta$ a covering pair. To highlight such a situation we write $\alpha \prec \beta$ for short. Whenever $\alpha$ is meet-irreducible (join-irreducible) congruence, then there is a unique congruence $\alpha^{+}$ ($\alpha^{-}$) such that $\alpha \prec \alpha^{+}$ ($\alpha^{-} \prec \alpha$). 
For a nilpotent algebra $\m A$  from CM wheneverver its congruencess $\alpha, \beta$ satisfy $\alpha \prec \beta$, the cosets (congruence classes) of $\beta/\alpha$ in $\m A/\alpha$ have equal sizes, being a power of some prime (see \cref{lm:simple-atom} and \cref{lm:simple-module-atom}). We later denote this prime by $\charr(\alpha, \beta)$ and call it a characteristic of a congruence cover $\alpha \prec \beta$. Moreover for arbitrary $\alpha < \beta$, we write $\charrset{\alpha, \beta}$ for the set of all possible prime characteristics of covering pairs, which are fully contained in $\intv{\alpha}{\beta}$.
We say that a pair of congruences $\alpha < \beta$ of a nilpotent algebra $\m A$ forms a Prime Uniform Product Interval (PUPI), if there are congruencess $\alpha < \alpha_1, \ldots, \alpha_k \leq \beta$ such that
\begin{itemize}
    \item $\bigvee \alpha_i = \beta$
    \item $\alpha_i \wedge (\bigvee_{j\neq i}\alpha_j) = \alpha$, for $i,j \in \{1..k\}$, 
    \item For every $i$  we have $|\charrset{\alpha, \alpha_i}| = 1$.
\end{itemize}

We call a congruence $\beta$ supernilpotent whenever it is nilpotent and the interval $\intv{0_{\m A}}{\beta}$ is a PUPI, and we call an algebra $\m A$ supernilpotent, whenever $1_{\m A}$ is supernilpotent. 
Each supernilpotent algebra is isomorphic to a direct product of nilpotent algebras of prime power size. In this sense supernilpotence generalizes nilpotence for groups.
Note that this definition is equivalent to a standard definition of supernilpotent algebras in congruence modular varieties \cite{MayrS21}.

 We say that a nilpotent algebra $\m A$ has supernilpotent rank $k$ whenever $k$ is the smallest number for which we can find sequence of congruences $0_{\m A} = \alpha_0 < \alpha_1 < \ldots < \alpha_k = 1_{\m A}$ such that interval $\intv{\alpha_i}{\alpha_{i+1}}$ is a PUPI for each $0 \leq i < k$.  Note that for a finite nilpotent algebra $\m A$ we can always find such a finite sequence, since each covering pair forms a PUPI. This notion of supernilpotent rank of an algebra, later denoted by $\sr{\m A}$, proved to be extremely usuful in some very recent results on the computation complexity of circuit satisfiability problem \cite{IdziakKK20, Kompatscher21}. In fact, results for $\progcsat{}/\ceqv{}$ we present in this paper, are essentially about nilpotent algebras with $\sr{A} = 2$.   
\section{Polynomial equivalence}
Our characterization of polynomial time cases of \poleqv{} is achieved through a reduction to some special instances of  \polsat{}. It was noticed already in \cite{GoldmannR02} that \polsat{\m G} for finite group $\m G$ reduces in polynomial time to \progpolsat{\m G}. Very recently, the full characterization,  of groups for which \progpolsat{} can be solved in randomized polynomial time was shown under the assumptions of \rethh and \cdhh in
\cite{IdziakKKW22-icalp}. 

 \begin{thm}
\label{thm:progsat}
Let $\m G$ be a finite group.
Assuming \rethh and \cdhh
the problem
\progpolsat{\m G} is in  \rptime if and only if
$\m G/\m G_p$ is nilpotent for some
normal $p$-subgroup $\m G_p$ of $\m G$
(with $p$ being prime).
\end{thm}

\cref{thm:progsat} together with a result from \cite{{IdziakKKW22TOCS}} provides us with enough information to prove the \cref{eqv-groups}.

\begin{proof}
We start with recalling that \cite[Theorem 1]{IdziakKKW22TOCS} tells us that under \ethh,
$\poleqv{\m G} \in \ptime$ forces $\m G$ to have a normal nilpotent subgroup $\m H$
with nilpotent quotient $\m G/\m H$.
The very same proof actually gives the same structure of $\m G$ under
\rethh and $\poleqv{\m G} \in \rptime$.

For the converse we start by observing that
the nilpotent normal subgroup $\m H$ of $\m G$ is a product of its Sylow subgroups,
say $\m H_1,\ldots,\m H_s$ with $\m H_j$ being a $p_j$-group.
Each such subgroup $\m H_j$ is isomorphic to the quotient $\m H/\m H'_j$
for the normal subgroup $\m H'_j$ of $\m H$ consisting of all elements in $H$
with the order not divisible by $p_j$.
In particular $H'_1\cap\ldots\cap H'_s=\set{1}$.
Moreover for an inner automorphism $h$ of $\m G$ we have not only $h(H)=H$
but also $h(H'_j)=H'_j$, as $h$ has to preserve the order of elements.
This means that $\m H'_j$ is normal also in $\m G$.

Now we know that the quotient $\m G/\m H'_j$ has a nilpotent normal subgroup $\m H/\m H'_j$
with a nilpotent quotient $(\m G/\m H'_j)/(\m H/\m H'_j)=\m G/\m H$.
Thus \cdhh and Theorem \ref{thm:progsat} give us that\break $\progpolsat{\m G/\m H'_j}$ and consequently $\polsat{\m G/\m H'_j}$  are in \rptime.
Consequently\break $\poleqv{\m G/\m H'_j}\in \rptime$ as deciding whether $\po t=\po s$ holds in
$\m G/\m H'_j$ reduces to check if none of the $\card{G/H'_j}-1$ equations of the form $\po t\po s^{-1}=a$,
with $a\in G/H'_j\setm\set{1}$ has a solution.

Finally, $H'_1\cap\ldots\cap H'_s=\set{1}$ tells us that an equation holds in $\m G$
iff it holds in all the quotients $\m G/\m H'_j$, so that
$\poleqv{\m G}\in\rptime$.
\end{proof}

\section{Program satisfiability}
\label{section:progcsat}
The goal of this section is to prove \cref{thm:cm-progcsat}. First we observe in \cref{fact:red-to-progcsat} that if a nilpotent algebra has a Malcev term then \csat{} and \ceqv{} for this algebra reduces to \progcsat{}. Thus, intractability of \csat{} or \ceqv{} implies intractability of \progcsat{} and conversely if \progcsat{} is in \rptime, so are \csat{} and \ceqv{}.

\begin{fact}
\label{fact:red-to-progcsat}
For a finite nilpotent Malcev algebra $\m A$ the problems $\csat{\m A}$ and $\ceqv{\m A}$
are Turing reducible to $\progcsat{\m A}$.
\end{fact}

\begin{proof}
To prove the fact we start with fixing $a_0\in A$
and enumerating $A\setm\set{a_0}=\set{a_1,\ldots,a_k}$ to form $A_j=\set{a_0,a_j}$.
Using Malcev term $\po d$ we define the $k$-ary polynomial $\po f$ by putting
\[
\po f(x_1,\ldots,x_k)=\po d(\ldots\po d(\po d(x_1,a_0,x_2),a_0,x_3)\ldots,a_0,x_k).
\]
Obviously $\po f(a_0,\ldots,a_0)=a_0$.
To see that the other $a_j$'s are in $\po f(A_1,\ldots,A_k)$
simply evaluate $x_j$ by $a_j\in A_j$ and the rest of the $x_i$'s by $a_0$.

Observe that if $\m A$ is a nilpotent algebra with a Malcev term $\po d(x,y,z)$
and $\zero,a,b\in A$ then, by \cite[Lemma 7.3]{fm} we have
$a=b$ iff $\po d(a,b,\zero)=\zero$.
This immediately shows that in nilpotent Malcev algebras only equations of the form
$\po t(\o x)=\zero$ (i.e. equations in which one side is a constant polynomial)
need to be considered in satisfiability or equivalence problems.

Thus we start with an instance of \csat{} [or \ceqv{}] $\po t(\o x)=\zero$ and define $kn$-ary polynomial
\[
\po t'(x_1^1,\ldots,x_1^k,\dots, x_n^1,\ldots,x_n^k)
=\po t(\po f(x_1^1,\ldots,x_1^k),\ldots,\po f(x_n^1,\ldots,x_n^k)).
\]
Due to $A = \po f(A_1,\ldots,A_k)$ we easily get that $\po t=e$
has a solution in $\m A$ [holds identically in $\m A$] iff
$\po t'=\po e$
has a solution with $x_i^j$ restricted to be taken from $A_j$
[or holds for all $2^{kn}$ evaluations of the $x_i^j$'s in $A_j$, respectively].

We define the reduction which for a given instance of \csat{\m A} [$\ceqv{\m A}$] in the form $\po t(\o x)=e$ returns  $nk$-ary boolean $\po t'$-program
(over the variables $\b_1^1,\ldots,\b_1^k,\ldots,\b_n^1,\ldots,\b_n^k$) with the instructions
$\iota(x_i^j)=(\b_i^j,a_0,a_j)$. One can easily see that these reductions work, after setting the accepting values
to be $S=\set{\zero}$ in case of $\csat{\m A}$,
and $S=A\setm\set{\zero}$ for $\ceqv{\m A}$.
\end{proof}
In the proof of Theorem \ref{thm:cm-progcsat} we will use Fact \ref{fact:red-to-progcsat} to show that, under our assumptions, nilpotent Malcev algebra with tractable \progcsat{} has supernilpotent rank equal at most $2$. In section \ref{section:hard} we use advanced tools of Tame Congruence Theory and Commutator Theory to prove the following lemma which shows that not for every algebra with supernilpotent rank equal $2$ \progcsat{} is tractable.

\begin{lm}
\label{lm:hard}
Let $\m A$ be a finite nilpotent algebra from a \cm variety with $\sr{\m A}=2$ and tractable $\progcsat{\m A}$.

Then $\m A$ has a supernilpotent congruence $\alpha$ with cosets of prime power order such that quotient algebra $\m A/\alpha$ is also supernilpotent, or \rethh fails.
\end{lm}

The important ingredient of the proof of above lemma is an idea of Barrington et al \cite{BarringtonBR94} heavily explored in \cite{idziakKK22LICS} and \cite{IdziakKKW22-icalp} resulting in the following lemma.

\begin{lm}
\label{lm:pseudo-and}
Let $p$ be a prime number and $\nu \geq 1$ be an integer.
Then for each 3-CNF formula $\Phi(\o x)$ with $n$ variables
there is a polynomial $w^\Phi_{p}(\o x)$ over  $GF(p)$
of degree at most $O(p^\nu)$
such that for all $\o b \in \set{0,1}^n$ we have
\[
w^\Phi_{p}(\o b) =
\left\{
\begin{array}{ll}
0, &\mbox{if the number of unsatisfied (by $\o b$) clauses in $\Phi$}\\
    &\mbox{is divisible by $p^\nu$}\\
1, &\mbox{otherwise.}
\end{array}
\right.
\]
Moreover, computing $w^\Phi_{p}$ from $\Phi$ can be done in $2^{O(p^\nu(\log n+\log p))}$ steps.
\end{lm}

The power of \cref{lm:pseudo-and} can be observed when we use it simultaneously for two different primes, say $p_1$ and $p_2$. Then if for a given 3-CNF formula $\Phi$ with $m$ clauses we will choose  positive integers $\nu_1$, $\nu_2$ such that $p_i^{\nu_i-1}\leq \sqrt{m}< p_i^{\nu_i}$, we will get, by Chinese Remainder Theorem, that $\Phi$ is satisfied by $\o b\in\set{0,1}$ iff $w^\Phi_{p_1}(\o b) =w^\Phi_{p_2}(\o b) = 0$. Moreover, the lengths of $w^\Phi_{p_1}(\o b)$ and $w^\Phi_{p_2}(\o b)$ are subexponential in the size of $\Phi$. The core of the proof of \cref{lm:hard} is showing (with haevilly use of Tame Congruence Theory and Commutator Theory) that if nilpotent Malcev algebra $\m A$ with supernilpotentn rank $2$ has supernilpotent congruence $\alpha$ which cosets are not of prime power size and such that $\m A/\alpha$ is supernilpotent then we can simulate by programs over $\m A$ systems of equations in the form:
\[
w^\Phi_{p_1}(\o b) = 0,
\]
\[
w^\Phi_{p_2}(\o b) = 0.
\]

Now we are ready to prove \cref{thm:progsat}
\begin{proof}[Proof of Theorem \ref{thm:progsat}:]
We start with the following observation:
\begin{senumerate}
\item
\label{fact:polcsat-lat}
$\progcsat{(\set{0,1};\meet,\join)}$ is \npc.
\end{senumerate}
To see that we start with an $n$-ary CNF-formula $\Phi(b_1,\ldots,b_n)$,
treat it as a function $\bool^n \map \bool$, and convert to $n$-ary program over the lattice
$(\set{0,1};\meet,\join)$.
First, by introducing the variables $x_1,\ldots,x_n$ and $x'_1,\ldots,x'_n$
we produce a new $2n$-ary formula $\Phi'(x_1,\ldots,x_n,x'_1,\ldots,x'_n)$
by simply replacing each positive literal $b_i$ by $x_i$ and negative literal $\neg b_i$ by $x'_i$.
This leads to a function $\Phi':\set{0,1}^{2n} \map \set{0,1}$
and allows us to transform the formula $\Phi$
to the program $\progg{\Phi'}{n}{\iota}{\set{1}}$
by putting $\iota(x_i)=(b_i,0,1)$ and $\iota(x'_i)=(b_i,1,0)$.

\medskip
\noindent
Using (\ref{section:progcsat}.\ref{fact:polcsat-lat}) we can exclude TCT types $\tn 3$ and $\tn 4$
from the typeset $\typset{\m A}$ so that:
\begin{senumerate}
\item
\label{solvable}
Either $\progcsat{\m A}$ is \npc or $\m A$ is solvable.
\end{senumerate}

\noindent
Actually we can force $\m A$ to be nilpotent.
Indeed, Lemma 2.2 of \cite{IdziakKK22STACS} supplies us with an element $\ee\in A$ and
a partition of $A$ into two nonempty disjoint subsets $A=A_0 \cup A_1$
which allows to associate (in linear time $O(m)$)
with a $3$-CNF-formula $\Phi$ (with $m$ clauses and $n$ variables)
a $3m$-ary circuit(polynomial) $\sacik_\Phi$ of $\m A$
such that for $b_1^1,b_2^1,b_3^1,\ldots,b_1^m,b_2^m,b_3^m\in\set{0,1}$
and $x_1^1,x_2^1,x_3^1,\ldots,x_1^m,x_2^m,x_3^m\in A$ with $x^j_i\in A_{b^j_i}$
we have
\[
    \Phi(b_1^1,b_2^1,b_3^1,\ldots,b_1^m,b_2^m,b_3^m)=1
    \quad\mbox{iff}\quad
    \sacik_\Phi(x_1^1,x_2^1,x_3^1,\ldots,x_1^m,x_2^m,x_3^m) = e.
\]
Thus fixing $a_0\in A_0$ and $a_1\in A_1$
we end up with a program $\progg{\sacik_\Phi}{3m}{\iota}{\set{e}}$ over $\m A$,
where $\iota(x_i^j)=(b_i^j,a_0,a_1)$.
This reduction from 3-CNF-SAT to $\progcsat{\m A}$ shows that:
\begin{senumerate}
\item
\label{nilpotent}
Either $\progcsat{\m A}$ is \npc or $\m A$ is nilpotent.
\end{senumerate}

\noindent
To enforce that tractability of $\progcsat{\m A}$ enforces $\m A$ to have supernilpotent rank $2$
we refer to \cite{IdziakKK20}, where $\sr{\m A}\geq 3$ gives a chain $p_1\neq p_2\neq p_3$
of primes occurring as characteristics in consecutive \pupi in $\con{\m A}$.
Moreover \cite{IdziakKK20} shows how to use one alternation of characteristics
to represent $n$-ary $\ccand$ by a polynomial/circuit of size $2^{O(n)}$,
and further how to use two alternations of characteristics
to compose such created $\sqrt{n}$-ary $\ccand$ functions
to represent $n$-ary $\ccand$ by a polynomial/circuit of size $2^{O(\sqrt{n})}$.
From this one expects that under \prethh $\csat{\m A}$ is not in \prptime if $\sr{\m a}\geq 3$.
In fact \cite{IdziakKK20} provides examples of such algebras,
while \cite{Kompatscher21} contains a nice proof of this expectation. This, together with \cref{fact:red-to-progcsat}, gives us that:

\begin{senumerate}
\item
\label{2-supernil}
If $\progcsat{\m A}\in \prptime$ then $\sr{\m A}\leq 2$, or \prethh fails.
\end{senumerate}
Now, the immediate consequence of (\ref{section:progcsat}.\ref{2-supernil}) and \cref{lm:hard} is that:

\begin{senumerate}
\item
\label{2-supernil-and-one-prime}
If $\progcsat{\m A}\in \prptime$ then there exists supernilpotent congruence $\alpha$ of $\m A$ with costes of prime power size and such that $\m A/\alpha$ is supernilpotent, or \prethh fails.
\end{senumerate}

Finally, let $\m A$ be a nilpotent algebra of supernilpotent rank $2$ having supernilpotent congruence $\alpha$ of $\m A$ with costes of prime power size and such that $\m A/\alpha$ is supernilpotent. In such the case  programs over $\m A$ computing $n$-ary $\ccand$ functions have size at least $2^{\Omega(n)}$, or \prethh and \cdhh fails.
This is an immediate consequence of \cref{thm:2supernil-circuit-early},
for if not, the $\ccand_n$ function computed by a program of size $\ell(n)$
can be also computed by a $\ccand_d\circ\ccmod_m\circ\ccmod_p$-circuit of size $O(\ell(n)^c)$
for the constants $c,d,m,p$ depending only of the algebra $\m A$.
But then CDH tells us that $O(\ell(n)^c)$, and therefore $\ell(n)$ itself,
has to dominate $2^{\Omega(n)}$.
Now, to prove that $\progcsat{\m A}$ is in \rptime recall that the second part of \cite[Proposition 9]{IdziakKKW22-icalp} tells
that a set of binary words accepted by a program $\progg{\po p}{n}{\iota}{S}$ of size $\ell$ is either empty
or has the size at least $2^n/\ell^{c}$, for some constant $c$. 
Thus checking at least $\ell^{c}$ random boolean words of length $n$ finds a word accepted by the program (if there is one at all)
with probability at least $1/2$.
This obviously puts $\progcsat{\m A}$ into \rptime.
\end{proof}

\section{Circuit equivalence}
In this section we will prove \cref{thm:cm-ceqv}. 
To do it we will show that solving \ceqv{} for an algebra $\m A$ can be reduced to solving the very same problem for quotients of $\m A$ by a meet-irreducible congruences. Obviously, every quotient algebra by a meet-irreducible congruence has the smallest congruence bigger than identity relation. This observation plays a crucial role in the proof o \cref{thm:cm-ceqv}.

\begin{proof}[Proof of \cref{thm:cm-ceqv}]
First suppose that $\ceqv{\m A}$ is tractable
to eliminate types $\tn 3$ and $\tn 4$ from the $\typset{\m A}$.
It should be obvious how to do it for type $\tn 3$.
For the lattice type $\tn 4$ we first note that the existence of a solution
(in the two element lattice)
to the systems of two equations of the form $\po m(\o x)=1 \ \& \ \po j(\o x)=0$ is \npc,
where $\po m(\o x)$ is in CNF and $\po j(\o x)$ is in DNF, both with only positive literals.
But obviously this system has no solutions iff
$\po m(\o x)\join\po j(\o x)=\po j(\o x)$ holds identically in $(\set{0,1};\meet,\join)$.
To put this into the minimal set $U$ of type $\tn 4$ simply replace each variable $x$ by $\po e_U(x)$
(where $\po e_U$ is an idempotent polynomial of $\m A$ with the range $U$)
and the operations $\meet,\join$ by the corresponding polynomials of $\m A$
turning $\m A|_U$ into the two element lattice.
This shows that $\typset{\m A}\ci \set{\tn 2}$ (i.e. $\m A$ is solvable) or $\ceqv{\m A}$ is \conpc.

Now the discussion made in \cite[Section 4]{IdziakKK22STACS} between Problems 2 and 3
shows that (under \prethh) in fact $\m A$ has to be nilpotent and that $\sr{\m A}\leq 2$.
This shows the `only if' direction of our theorem.

\medskip
To prove the converse note that an identity holds in an algebra $\m A$
iff it holds in all quotients of $\m A$ by meet-irreducible congruences
(as the intersection of all meet-irreducible congruences is the identity relation $\m 0_{\m A}$). Let $\alpha$ be a meet-irreducible congruence of $\m A$. To complete the proof it suffices to show that (under \cdhh) $\ceqv{\m A/\alpha}\in\rptime$. 

Observe that if $\m A$ is nilpotent and $\sr{\m A}\leq 2$
then each quotient of $\m A$ has these two properties as well. Moreover, identity relation in $\m A/\alpha$ (i.e. $0_{\m A/\alpha }$)  has the unique cover. Since congruence classes for covering pairs have equal sizes \cite[Colloraly 7.5]{fm}, it is clear from definition of supernilpotency that every supernilpotent congruences of $\m A/\alpha$ has cosets of prime power sizes (as for every $\beta_1,\beta_2\in\con{A/\alpha}$, if $\beta_1,\beta_2>0_{\m A/\alpha}$, then  $\beta_1\meet \beta_2> 0_{\m A/\alpha}$). Therefore by \cref{thm:cm-progcsat}, $\progcsat{\m A/\alpha}\in\rptime$. Finally \cref{fact:red-to-progcsat} gives us $\ceqv{\m A}\in\rptime$, as required.
\end{proof}

\section{Notation}

In this section we introduce detailed notation, needed in further sections of the paper.
\subsection*{Algebra} 
We use the standard universal algebraic notation the reader can find e.g. in \cite{BurrisSankappanavar}. Our results heavily rely on Tame Congruence Theory and Modular Commutator Theory which have a detailed description in \cite{hm} and \cite{fm} respectively. One can find the not too long summary of needed notions and facts  in \cite[Section 2]{IdziakK22}. Here we just recall, for readers convenience, the most important notation.

 For an algebra $\m A$ and $a,b\in A$ congruence $\Theta(a,b)$ is the smallest congruence containing the pair $(a,b)$. Note that every join-irreducible congruence is generated by a single pair of elements of the algebra. In particular, covers of the identity relation, called atoms, are generated by a single pair. 

 Minimal sets are the central notion of the Tame Congruence Theory. Formally, for an algebra $\m A$ and $\alpha,\beta\in \con {\m A}$ such that $\alpha < \beta$ we define the $(\alpha,\beta)$-minimal set $U\subseteq A$ as minimal, with respect to the inclusion, among sets of the form $\po f(A)$ for all unary polynomials  $\po f$ of  $\m A$ such that $\po f(\beta)\not\subseteq\alpha$. 
 In this paper we consider $(\alpha,\beta)$-minimal sets for $\alpha\prec\beta$ only (i.e. for so-called prime quotients).  For every minimal set $U$ of $\m A$ there exists unary idempotent polynomial $\po e_U$ of $\m A$ such that $\po e_U(A)=U$. A trace of the $(\alpha,\beta)$-mimal set $U$ is a set $u/\alpha\cap U$ for some $u\in U$ such that $u/\alpha\cap U\not=u/\beta\cap U$. For $(\alpha, \beta)$-minimal set $U$ (or its trace $N$) of algebra $\m A$ we usually consider an induced algebra $\m A/\alpha|_{U/\alpha}$ ($\m A/\alpha|_{N/\alpha}$). The five ways of local behaviour of finite algebras mentioned in \cref{sec-notions-1} are exactly the types of algebras (or rather their polynomial clones) induced on minimal sets traces. Note that all traces of a minimal set are polynomially equivalent. In this paper we usually work with minimal sets of type $\tn 2$, that is in which algebras induced on traces are polynomially equivalent to one-dimensional vector spaces. Traces of a minimal set of type $\tn 2$ are of prime power order, this prime is called a characteristic of the minimal set. Note that all minimal sets taken with respect to the same prime quotients induce polynomially equivalent algebras. It allows us to define the type of prime quotient as a type of minimal sets taken with respect to that quotient. In case of prime quotients of type $\tn 2$ a prime characteristic of the quotient defined in \cref{sec-notions-1} is equal to the characteristic of minimal sets taken with respect to it.

Algebras belonging to congruence modular variety are a wide class of algebras for which Tame Congruence Theory and Commutator Theory work particularly well. In this context term ,,modular'' means that all algebras from the variety have modular congruence lattice i.e. for $\alpha,\beta, \gamma\in \con{A}$ $\alpha\leq\beta$ implies $\alpha\join(\gamma\meet\beta)=(\alpha\join\gamma)\meet\beta$.  Equivalently, lattice is modular if it has no elements $\alpha,\beta,\gamma$ for which $\alpha<\beta$, $\alpha\join\gamma=\beta\join\gamma$ and $\alpha\meet \gamma=\beta\meet \gamma$. Such a sublattice $\{\alpha, \beta, \gamma, \alpha\join\gamma, \alpha\meet\gamma\}$, if found in $\con{A}$, is called a \textit{pentagon}.   If $\intv{\alpha}{\beta}$ and $\intv{\gamma}{\delta}$ are intervals such that $\beta\meet\gamma = \alpha$ and
$\beta\join \gamma = \delta$, then $\intv{\alpha}{\beta}$ is said to transpose up to $\intv{\gamma}{\delta}$, written $\intv{\alpha}{\beta}\nearrow\intv{\gamma}{\delta}$ and $\intv{\gamma}{\delta}$ is said to transpose down to $\intv{\alpha}{\beta}$, written $\intv{\gamma}{\delta}\searrow\intv{\alpha}{\beta}$ and the two intervals are called transposes of one another. Two intervals are said
to be projective if one can be obtained from the other by a finite sequence of transposes. A fundamental fact in lattice theory is that a lattice is modular if and only if its projective intervals are isomorphic. In the case of congruence lattice of an algebra this isomorphism is even stronger since projective prime  quotients have exactly the same minimal sets and in a consequence the same types and, in case of type $\tn 2$, the same characteristics.

Malcev algebras are examples of algebras from congruence modular varieties. Such algebras form so called congruence permutable varieties. Congruences of algebras from this class commute i.e.\! for arbitrary congruences $\alpha$ and $\beta$ we have that  $\alpha\circ\beta=\beta\circ\alpha$.   

Supernilpotency plays a crucial role in our investigations. For a fixed algebra $\m A$ we write $\sigma$ for a biggest supernilpotent congruence and $\kappa$ for the smallest congruence such that $\m A/\kappa$ is supernilpotent. Moreover, we write $\sigma_p$ for the biggest supernilpotent congruence with cosets of size being a power of $p$. The existance of such biggest/smallest congruences follows from term definition of supernilpotence and its properties \cite[Corollary 6.6]{aichmud-2010}. 

\subsection*{Arithmetic operations and functions}
Sometimes we use addition in two different abelian groups  inside one formula. To avoid ambiguity, different symbols for different types of addition are used: $+, \sum$  for one type of addition (for instance in $\z_p$) and $\oplus, \bigoplus$ for the other (in $\z_m$). Moreover, for a natural number $m$ with a prime decomposition $p_1^{\alpha} \cdot p_2^{\alpha_2} \cdot \ldots \cdot p_s^{\alpha_s}$ we write $\sdiv m=\set{p_1,\ldots,p_\s}$ for the set of prime divisors of $m$ and $\pdiv m = p_1 \cdot \ldots \cdot p_s$ for the largest square-free divisor of $m$. We also use $\sdiv A$ for $\sdiv {|A|}$ and $\pdiv A$ for $\pdiv {|A|}$. Likewise, $\ar A$ denotes the maximal arity among basic operations of $\m A$ and $\ar f$ is the arity of the function $f$. For a set $X=\set{d_1,\ldots,d_s}$
and a function $f: X^k \map Y$ we associate to $f$ its
binary expansion $\fcirc{f} : \bool^{\card{X}\cdot k} \map Y$,
i.e., arbitrary function satisfying
\[
f(x_1,\ldots,x_k) =
\fcirc{f}(x_1 \beq d_1,\ldots,x_1 \beq d_s,\ldots,x_k \beq d_1,\ldots,x_k \beq d_s).
\]

\subsection*{Programs}
Recall that a program $(\po p, n, \iota, S)$ computes a Boolean function $\{0,1\}^n \rightarrow \{0,1\}$. From now on we write $\progb{\po p}{\iota}{S}(\o b)$ for evaluation of this function on a tuple $\o b \in \{0,1\}^n$. Moreover, if $S = \{c\}$ is one-element set, we simply write $\progb{\po p}{\iota}{\ccc}(\o b)$ instead of $\progb{\po p}{\iota}{\{\ccc\}}(\o b)$. With each program $(\po p, n, \iota, S)$ over $\m A$ we can naturally associate an inner function $\prog{\po p}{\iota}: \{0,1\}^n \rightarrow A$ satisfying $\prog{\po p}{\iota}(\o b) \in S$ iff $\progb{\po p}{\iota}{S}(\o b) = 1$. This function is computed by the expression 
$\po p(a^{x_1}(b^{x_1}),\ldots,a^{x_k}(b^{x_k}))$, where all $a^{x_i}, b^{x_i}$ are provided by the instructions $\iota(x) = (b^x, a^x)$. For a congruence $\sigma \in \con{\m A}$ and a program $(\po p, n, \iota, S)$ we can define the quotient program $(\po p/\sigma, n, \iota/\sigma, S/\sigma)$ of an algebra $\m A/\sigma$, by simply reinterpreting $\po p$ to be a circuit over $\m A/\sigma$ (algebras $\m A$ and $\m A/\sigma$ have the same signature), and taking $\iota(x) = (b^x, a^{x}/\sigma)$, and $S/\sigma = \{[s]_{\sigma}: s\in S\}$.
\subsection*{Circuits} In next chapters we use different types of gates to build bounded-depth circuits computing Boolean functions. For instance we write $\ccand_{d}$ to denote a gate which takes at most $d$ inputs and computes their conjunction and $\ccor_{d}$ for a gate computing at most $d$-ary disjunction. A $\ccmod_m$-type boolean gate is any unbouded fan-in gate, which sums the inputs modulo $m$ and returns $1$ iff the sum belongs to some accepting set $S\subseteq \z_m$. We allow different accepting sets for different gates. Other kinds of gates appearing in the paper are $\sumpk{p}{\nu}$ gates which also rely on modulo counting, but in a more complex way. Each such a gate takes $n$ inputs $b_1, \ldots, b_n$, and computes their affine combination $\alpha_1 b_1 + \ldots + \alpha_n b_n + d$  in $\z_p^v$. Here we interpret each $b_i$ as $v$-dimensional vector $(b_i, \ldots, b_i)$. We allow each $\alpha_i$ to be arbitrary endomorphism of the abelian group $\z_p^v$, which can be also viewed as arbitrary $v \times v$ matrix with coefficients from $\z_p$. Hence $\alpha_i b_i$ can be viewed as applying linear map $\alpha_i$ to $(b_i, \ldots, b_i)$.  In this way $\sumpk{p}{\nu}$ gate computes a function of type $\{0,1\}^n \rightarrow \z_{p}^{\nu}$. Additionally, for $c\in \z_p$ let $\sumpk{p}{\nu}^c$ denote a boolean variant of $\sumpk{p}{\nu}$ gate which returns value $1$ when the affine combination $\alpha_1 b_1 + \ldots + \alpha_n b_n + d$ evaluates to $c$, and returns $0$ otherwise. Note that for a field $\mathbb{F}$ with underlying group $\z_p^{v}$, every scalar $g \in F$ defines a linear map $x \mapsto g \cdot x$. Thus such $g$ (or rather its associated endomorphism) can be used as a coefficient $\alpha_i$ inside $\sumpk{p}{\nu}$ gate.  To simplify some of the later calculations, we will always assume that vector $(1,\ldots, 1)$ is the unit of such a field $\mathbb{F}$.  Note that for a field $\mathbb{F} = (F, \cdot, +, 0, 1)$, for each $e \in F$ we can define a new field $\mathbb{F}_e = (F, \circ, +, 0, e)$, by defining $x \circ y$ to be $x \cdot e^{-1} \cdot y $.

It is indeed a valid assumption, since for every  non-zero element $g\in \z_p^{v}$ which is not a unit of the field $\mathbb{F}$ over $\z_p^{v}$, we can redefine multiplication to be $x \cdot g^{-1} \cdot y$ in terms of old multiplication, and fix the new inverse of an element $x$ to be $x^{-1}\cdot g^{2}$. One can easly check that such a rewriting of $\mathbb{F}$ defines a new field in which $g$ plays role of the unit (and we did not alter the underlying group).

Having all these gates, we build a bounded-depth circuits by listing types of gates which are allowed on each layer, starting with the input layer on the left, finishing with the output layer on the right. For instance $\ccand_d\circ\ccmod_m\circ\ccmod_p$ denotes a $3$-level circuit with inputs wired to some set of $\ccand_d$ gates, then this  $\ccand_d$ gates have wires to $\ccmod_m$ gates and on the output level there is one $\ccmod_p$ gate. We allow multiple wires between two gates.

\section{Useful algebraic facts}
\label{section:collection}

Before dealing with the remaining parts of the proofs, we present a number of lemmas describing certain useful aspects of the local behaviour of algebras.

\begin{lm}
\label{lm:simple-atom}
Let $\m A$ be a finite algebra and $\beta$ be an atom in $\con A$.
Then for each element $\zero\in A$ the induced algebra $\m A|_{\zero/\beta}$
on the coset $\zero/\beta$
is either trivial or simple.
\end{lm}
\begin{proof}
We show that collapsing in $\m A|_{\zero/\beta}$ two elements $a\neq b$ we collapse any pair
$(c,d)\in (\zero/\beta)\times(\zero/\beta)$.
Obviously we have $\Theta_{\m A}(a,b)=\beta$ so that $c$ and $d$ can be connected by the so-called Malcev chain, i.e. the projections of the pair $(a,b)$ by unary polynomials of $\m A$.
More precisely it means that there are unary polynomials $\po p_0,\ldots,\po p_s$ of $\m A$ with
$c\in\set{\po p_0(a),\po p_0(b)}$,
$\set{\po p_j(a),\po p_j(b)}\cap\set{\po p_{j+1}(a),\po p_{j+1}(b)}\neq\emptyset$ and
$d\in\set{\po p_s(a),\po p_s(b)}$.
This gives that each of the $\po p_j$'s maps $\zero/\beta$ into $\zero/\beta$,
i.e. $\po p_j$ is a basic operation of the induced algebra $\m A|_{\zero/\beta}$.
But then the very same chain is a Malcev chain inside $\m A|_{\zero/\beta}$
showing that $(c,d)\in \Theta_{\m A|_{\zero/\beta}}(a,b)$, as claimed.
\end{proof}

The next lemma is a specialized version of some much more general facts known in Universal Algebra (see for instance \cite[Corollary 5.8]{fm}).

\begin{lm}
\label{lm:simple-module-atom}
Let $\m A$ be a finite nilpotent algebra from a congruence modular variety
with a Malcev term $\po d(x,y,z)$
and $\beta$ be an abelian atom in $\con A$.
Then for each element $\zero\in A$ the induced algebra $\m A|_{\zero/\beta}$ on the coset $\zero/\beta$ is polynomially equivalent to a simple module in which the underlying group structure is determined by the binary operation $x+y=\po d(x,\zero,y)$
and the corresponding group $(\zero/\beta; +,\zero)$
is isomorphic to some power of the group $(\z_p;+,0)$.
\end{lm}


\begin{proof}
By \cite[Corollary 5.8]{fm} we know that the coset $\zero/\beta$ with the operation
$x+y=\po d(x,\zero,y)$ is a group.
Moreover, abelianity of the congruence $\beta$ gives that the induced algebra $\m A|_{\zero/\beta}$ is abelian, and therefore, by \cite[Corollary 5.9]{fm} is affine,
i.e. polynomially equivalent to a finite module $M$.
But \cref{lm:simple-atom} ensures us that this finite module is simple.
This in particular gives that all non-zero (i.e. different from $\zero$) elements of $M$ have the same order, as otherwise for $a,b\in M=\zero/\beta$ of orders $m,n$ respectively,
if $m<n$ then the endomorphism $M\ni x \mapsto mx\in M$ is non-zero (as $mb\neq \zero$)
and is not surjective (as $ma=\zero$), contradicting simplicity of $M$.
This shows that the underlying group of the module $M$ is a finite elementary abelian group,
i.e. a finite power of the group $(\z_p;+,0)$ for some prime $p$.
\end{proof}

Now we present yet another simple observation.

\begin{lm}
\label{lm:minset}
Let $\m A$ be a finite nilpotent Malcev algebra
and $\delta\prec\theta$ two of its congruences.
Then every element $e\in A$ belongs to some $(\delta,\theta)$-minimal set.
\end{lm}

\begin{proof}
Start with any $(\delta,\theta)$-minimal set $V$ and pick $(c,d)\in\theta|_V\setm\delta$.
By \cite[Corollary 7.5]{fm} in nilpotent algebras all cosets of the same congruence have the same size. Thus there is $a \in A$ with $(e,a)\in \Theta(c,d)\setm\delta$.
Since $\m A$ is Malcev it has the unary polynomial $\po p(x)$ such that
$\po p\vpair{c}{d} = \vpair{e}{a}$.
In particular the polynamial $\po p$ does not collapse $\theta$ to $\delta$
so that, by \cite[Lemma 2.8]{hm}, $\po p(V)$ is a minimal set and it contains our element $e$.
\end{proof}

Now we demonstrate that easiness of \progcsat{} can be transferred to quotients. 
\begin{fact}

\label{fact:quotient}
For an algebra $\m A$ and its congruence $\theta$
there is a polynomial time reduction from \progcsat{\m A/\theta} to \progcsat{\m A}.
\end{fact}

\begin{proof}
Let  $(\po p, n, \iota, S)$ be a program over $\m A/\theta$. Then for $\iota(x) = (b^x, [a_0^x]_{\theta}, [a_1^x]_{\theta})$ we put $\iota'(x) = (b^x, a_0^x, a_1^x)$ and $S' = \bigcup S$ (where $a_0^x$, $a_1^x$ are arbitrary representants of the respective cosets in $\m A/\theta$). Then since $\theta$ is a congruence relation, for any $\o b \in \{0,1\}^n$ we have $\prog{\po p}{\iota}(\o b) \in S \iff \prog{\po p}{\iota'}(\o b) \in S'$. So the program $(\po p, n, \iota', S')$ over $\m A$ computes exactly the same Boolean function as $(\po p, n, \iota, S)$ over $\m A/\theta$, while both programs clearly have the same size. Obviosuly $(\po p, n, \iota', S')$ can be computed from $(\po p, n, \iota, S)$ in polynomial (linear) time.

\end{proof}

The following simple fact provides us with a way to reason about characteristics of intervals below a join irreducible congruence in the congruence lattice of an algebra.

\begin{fact}\label{fact:ji-PUPI}
Let $\m A$ be a solvable Malcev algebra and $\alpha,\beta\in \con{A}$ such that $\intv{\alpha}{\beta}$ is a \pupi.

Then,  for every $\gamma\in\con{A}$ join irreducible in $\intv{\alpha}{\beta}$ it follows that $|\charrset{\alpha,\gamma}|=1$.
\end{fact}
\begin{proof}
From the fact that $\intv{\alpha}{\beta}$ is a \pupi we immediately have that $\beta/\alpha$ and hence $\gamma/\alpha$ are supernilpotent in $\m A/\alpha$. Therefore, $\intv{\alpha}{\gamma}$ is a \pupi. Hence, $\gamma=\bigvee_{i\in I}\alpha_i$ where $\{\alpha_i\}_{i\in I}$ is a set of maximal congruences from  $\intv{\alpha}{\gamma}$ with $|\charrset{\alpha,\alpha_i}|=1$. But $\gamma$ is join-irreducible in  $\intv{\alpha}{\beta}$, and hence in  $\intv{\alpha}{\gamma}$, and  $\gamma=\bigvee_{i\in I}\alpha_i$ implies that $\gamma=\alpha_i$ for some $i$. Thus $|\charrset{\alpha,\gamma}|=1$.
\end{proof}

\section{ProgramCSat -- hardness}
\label{section:hard}

The next Lemma is modeled after  Lemma 3.1 from \cite{ikk:mfcs}. It provides us with a normal form for all the $s$-ary functions of type $\z_m^s \map \z_p$.

\begin{lm}
\label{lm-zpqe}
Let $m$ be a square-free positive integer and $p\nmid m$ be a prime.
Then every function $f: \z_m^s \map \z_p$ can be expressed by
\[
f(x_1,\ldots,x_s)=\sum_{(\o\beta,u)\in\m \z_m^s\times\z_m}
\mu_{\o\beta,u}\cdot \b\left(\bigoplus_{i=1}^s \beta_i x_i\oplus u\right),
\]
where
\begin{itemize}
    \item $\b:\z_m\map\z_p$ is given by $\b(0)=1$ and $\b(x)=0$ for $x \neq 0$,
    \item $\mu_{\o\beta,u}\in \z_p$,
    \item $\sum$ is the addition from $\mathbb{Z}_p$,
    \item $\bigoplus$ is the addition from $\mathbb{Z}_m$.
\end{itemize}
Moreover, coefficients $\mu_{\o\beta,u}\in \z_p$ are computable in $2^{O(s)}$ steps.
\end{lm}
\begin{proof}
We start with assuming that $m$ is a prime, say $\sdiv{m}=\set{q}$, to recursively define
the functions $\b_k^q : \z_q^k\map\z_p$ by putting
\begin{align*}
\b_1^q(x_1) &= \b(x_1)\\
\b_{k}^q(x_1,\ldots,x_{k}) &=
q^{-1}\cdot\left(
\sum_{i=0}^{q-1} \b_{k-1}^q(x_1,\ldots,x_{k-2},x_{k-1}\oplus i\cdot x_{k}) -
\sum_{i=1}^{q-1} \b_{k-1}^q(x_1,\ldots,x_{k-2},x_{k}\oplus i\cdot 1_q)\right),
\end{align*}
where $q^{-1}$ is the inverse of $q$ in the field $GF(p)$.
Then we observe that
$\b_k^q(x_1,\ldots,x_{k})=1$ if $x_1=\ldots=x_{k}=0$ and $\b_k^q(x_1,\ldots,x_{k})=0$ otherwise.

To see this we induct on $k$ and first assume that $x_k\neq 0$.
Then in each of the two big sums only one summand can be non-zero,
i.e. the one in which the last argument $x_{k-1}\oplus i\cdot x_{k}$ (or $x_{k}\oplus i\cdot 1_q$, resp.)
is zero.
But then each big sum is $\b_{k-1}^q(x_1,\ldots,x_{k-2},0)$, so that $\b_k^q(x_1,\ldots,x_{k})=0$.
On the other hand for $x_k=0$ all summands in the first sum are equal to $\b_{k-1}^q(x_1,\ldots,x_{k-1})$.
Simultaneously all summand of the second sum are zero, as
$\b_{k-1}^q(x_1,\ldots,x_{k-2},x_{k}\oplus i\cdot 1_q)=\b_{k-1}^q(x_1,\ldots,x_{k-2},i\cdot 1_q)$
and $i\cdot 1_q$ ranges here over all non-zero elements of $\z_q$.
Thus
\[
\b_{k}^q(x_1,\ldots,x_{k}) =
q^{-1}\cdot\left(
\sum_{i=0}^{q-1} \b_{k-1}^q(x_1,\ldots,x_{k-1}) - 0\right) =
\b_{k-1}^q(x_1,\ldots,x_{k-1}).
\]

Now, for $m=q_1\cdot\ldots\cdot q_r$ first we observe that the group $\z_{q_j}$
is isomorphic to the subgroup $\frac{m}{q_j}\cdot\z_m$ of $\z_m$,
e.g. under sending the unit $1_{q_j}$ to $\frac{m}{q_j}\cdot 1_m$.
Now for $x_1,\ldots,x_{k}\in\z_m$ we put
$\b_k^j(x_1,\ldots,x_{k})=\b_k^{q_j}(\frac{m}{q_j}\cdot x_1,\ldots,\frac{m}{q_j}\cdot x_{k})$
to get
$\b_k^j(x_1,\ldots,x_{k})=1$ if $x_1\congruent{q_j}\ldots\congruent{q_j} x_{k}\congruent{q_j}0$
and $\b_k^j(x_1,\ldots,x_{k})=0$ otherwise.

Finally we put
\(
\b_k(\o x) = \prod_{j=1}^r \b_k^j(\o x)
\)
to get the functions $\b_k : \z_m^k\map\set{0,1}\ci\z_p$ satisfying
$\b_k(x_1,\ldots,x_{k})=1$ iff $x_i\congruent{q_j}0$ for all $i,j$.

Note that the recursive definition of the $\b_k^j$'s puts them in a nice form
\[
\sum_{(\o\beta,u)\in\m \z_m^k\times\z_m}
\mu_{\o\beta,u}\cdot \b\left(\frac{m}{q_j}\cdot\left(\bigoplus_{i=1}^s \beta_i x_i\oplus u\right)\right),
\]
i.e. they are $\z_p$-linear combinations of the expressions of the form $\b(A\cdot\o x)$ with
$A\cdot\o x$ being $\z_m$-linear combinations
of the $\frac{m}{q_j}\cdot x_i$'s and $\frac{m}{q_j}\cdot 1_m$.
Now, distributing these $\z_p$-linear combinations in the product of the $\b_k^j$'s
we end up with a $\z_p$-linear combinations of the expressions of the form
$\b(\frac{m}{q_1}\cdot A_1\o x)\cdot\ldots\cdot\b(\frac{m}{q_r}\cdot A_r\o x)$.
But this last product is in fact equal to
$\b(\frac{m}{q_1} A_1\o x \oplus\ldots\oplus\frac{m}{q_r} A_r\o x)$,
as $\z_m$ is the direct sum of the $\z_{q_j}'s$,
i.e. the subgroups of $\z_m$ generated by the $\frac{m}{q_j}$'s.

Obviously this final $\z_p$-linear combination of the
$\b(\frac{m}{q_1} A_1\o x \oplus\ldots\oplus\frac{m}{q_r} A_r\o x)$'s
can be computed in $m^{O(k)}$ steps, as nice representations of the $\b_k^j$'s
are computable in $q^{O(k)}$ steps.

\medskip
To provide a nice representation of an arbitrary function $f: \z_m^s \map \z_p$
simply take the $\z_p$-linear combinations of the $\b_s$ with appropriate coefficients,
i.e. observe that
\[
f(x_1,\ldots,x_s) =
\sum_{(a_1,\ldots,a_s)\in\z_m^s} f(a_1,\ldots,a_s)\cdot \b_s(x_1-a_1,\ldots,x_s-a_s).
\]
Again this can be easily done in $m^{O(s)} = 2^{O(s)}$ steps.
\end{proof}

\begin{proof}[proof of Lemma \ref{lm:pseudo-and}]
Suppose the formula $\Phi$ has $\ell$ clauses and $n$ variables.
The mentioned result of \cite{BarringtonBR94}, reformulated in \cite[Fact 3.4]{idziakKK22LICS},
endows us with an $\ell$-ary polynomial $w(c_1,\ldots,c_\ell)$ of the field $GF(p)$ of degree $p^{\nu}-1$
such that for $x\in\set{0,1}^\ell$ we have
\[
w(\o x) =
\left\{
\begin{array}{ll}
0, &\mbox{if the number of zero entries of $\overline{x}$ is divisible by $p^\nu$,}\\
1, &\mbox{otherwise.}
\end{array}
\right.
\]
Note, that we are interested only in the behavior of $w(\o x)$ on values from the set $\set{0,1}$. It means that we can assume that the highest power of the variable that occurs in the polynomial is one as for $x\in\set{0,1}$ we have that $x=x^2$. Hence, it is enough to consider polynomials in sparse form consisting of monomials of the form $\prod_{x\in V} x$, where $V$ is a subset of variables of size at most $p^\nu-1$. Now, observe that considering subsets of variables of size at most $p^\nu-1$ ordered by its sizes we can step by step determine coefficients for all monomials.
Indeed, coefficient $\alpha_V$ cooresponding to monomial $\prod_{x\in V} x$  can be computed from all $\alpha_{V'}$ for $V' \subsetneq V$ and the evaluation of $w$ which puts $1$'s precisely to the variables in $V$ (and $0$'s to other variables).
Since we consider at most $O(n^{p^\nu})=O(2^{p^\nu\log{n}})$ monomials in  which is also the upper bound for the size of $w$, we can compute $w$ in time bounded by $2^{O(p^\nu\log{n})}$.

All we need to do to get $w^\Phi_{p}$
is to substitute the variable $c_i$ in $w$ by a polynomial
$1-\hat{l}^i_1\cdot\hat{l}^i_2\cdot\hat{l}^i_3$ of degree $3$
that codes a clause $C_i=l^i_1\vee l^i_2\vee l^i_3$ of $\Phi$,
i.e. $\hat{l}^i_j$ is set to $1-x$ if $l^i_j=x$ and to $x$ if $l^i_j=\neg x$.
\end{proof}

\begin{lm}
\label{lm:beta-int}
Let $\m A$ be a finite nilpotent Malcev algebra,
$\alpha$ its join irreducible congruence with $\alpha^-$ being its unique subcover.
Moreover let $\beta\in \con A$ be a subcover of $\alpha^-$ with $\charr(\beta,\alpha^-)\neq\charr(\alpha^-,\alpha)$.
Then for any choice of $(c,d)\in\alpha\setm\alpha^-$ and $(e,a)\in\alpha^-\setm\beta$
every function $f: \set{c,d}^s \map \set{e,a}$ can be $\beta$-interpolated
by an $s$-ary polynomial $\po p$ of $\m A$,
i.e. $f(\o x)\congruent{\beta}\po p(\o x)$, whenever $\o x \in\set{c,d}^s$.
Moreover such a polynomial $\po p$ can be obtained from the function $f$ in $2^{O(s)}$ steps.
\end{lm}

\begin{proof}
To start with we fix a Malcev term $\po d(x,y,z)$ for $\m A$ and the characteristics
$p=\charr{(\beta,\alpha^-)}\neq\charr{(\alpha^-,\alpha)}=q$.

By \cite[Theorem 2.8.(4)]{hm} we can project the pair $(c,d)$ by a unary polynomial, say $\po g$,
of $\m A$ to a $(\alpha^-,\alpha)$-minimal set $U$
so that $(c',d')=(\po g(c),\po g(d))\in \alpha\setm\alpha^-$.
Note that the trace $N=c'/\alpha \cap U$ modulo $\alpha^-$ (i.e. $N/\alpha^-$)
is an elementary $q$-group with respect to $x\oplus y = \po d(x,c',y)$
and $c'/\alpha^-$ being its zero element.
This in particular tells us that the subgroup of $N/\alpha^-$ generated by $d'/\alpha^-$is isomorphic to $\z_q$.
We pick $c_0,c_1,\ldots,c_{q-1}\in A$ so that $c_0=c'$, $c_1=d'$ and
$c_0/\alpha^-|_U,c_1/\alpha^-|_U,\ldots,c_{q-1}/\alpha^-|_U$
lists all elements of this subgroup and put
$N'=\bigcup_{j\in\z_q} c_j/\alpha^-|_U$.

With the help of \cref{lm:minset} we pick a $(\beta,\alpham)$-minimal set $V$ containing $e$.
Again we know that the trace $M=e/\beta \cap V$ is, modulo $\beta$, an elementary $p$-group
with respect to $x+y = \po d(x,e,y)$.
Pick $b\in M\setm e/\beta$ and let $\po h$ be a unary polynomial of $\m A$ witnessing the fact that
$(e,b)\in \alpham \ci\alpha=\Theta(c',d')$, i.e. $\po h\vpair{c'}{d'} = \vpair{e}{b}$.
We may additionally assume that $\po h(A)\ci V$, or replace $\po h$ by $\po e_v\po h$, where $\po e_v$ is a unary idempotent polynomial with the range $V$.
Since $p\neq q$ we know that, due to \cite[Lemma 4.30]{hm},
the minimal set $U=\po e_U(A)$ can contain no $(\beta,\alpham)$-minimal set,
and therefore $\alpham|_U\ci\beta$.
By a similar reason the minimal set $V$, and therefore its subset $\po h(A)$,
cannot contain $(\alpham,\alpha)$-minimal set. Thus $\po h(\alpha)\ci\alpham$.
Summing up we know that
\begin{senumerate}
\item
\label{pol-h}
The algebra $\m A$ has a unary polynomial $\po h$ satisfying
\begin{itemize}
  \item $\po h(A)\ci V$,
  \item $\po h(c_0)\in e/\beta$,
  \item $(\po h(c_0),\po h(c_1))\in\alpham\setm\beta$,
  \item $\po h(\alpha)\ci \alpham$,
  \item $\po h(\alpham|_U)\ci \beta$.
\end{itemize}
\end{senumerate}
Now we note that the operations $\oplus$ and $+$ do not need to be associative or commutative in $A$, but they  do enjoy these properties, modulo $\alpham$ or $\beta$ respectively,
on the traces $N$ or $M$.

In the following we will need a polynomial $\po h$ that enjoys the properties listed in
(\ref{section:hard}.\ref{pol-h}) as well as
\begin{senumerate}
\item
\label{pol-h-sum}
$\sum_{j\in \z_q} \po h(c_j)\not\in e/\beta$.
\end{senumerate}
If our original $\po h$ does not satisfy (\ref{section:hard}.\ref{pol-h-sum}),
we simply replace it by
$\po h'(x)=\po h(x\oplus c_1)-\po h(c_1)$, so that
(\ref{section:hard}.\ref{pol-h}) still holds and
$\sum_{j\in \z_q}\po h'(c_j)=
\sum_{j\in \z_q}\left(\po h(c_j\oplus c_1)- \po h(c_1)\right) \congruent{\beta}
\left(\sum_{j\in \z_q}\po h(c_j)\right)-q\cdot\po h(c_1) \congruent{\beta}
-q\cdot\po h(c_1)\not\in e/\beta$.
To see the middle equality note that
$\set{(c_0\oplus c_1)/\alpham,\ldots,(c_{q-1}\oplus c_1)/\alpham}
=\set{c_0/\alpham,\ldots,c_{q-1}/\alpham}$
and consequently
$\sum_{j\in \z_q}\po h(c_j\oplus c_1)\congruent{\beta} \sum_{j\in \z_q}\po h(c_j)$,
by the last item in (\ref{section:hard}.\ref{pol-h}).

Since $a'=\sum_{j\in \z_q} \po h(c_j)\not\in e/\beta$
the subgroup of $(M/\beta,+)$ generated by $a'/\beta$ is isomorphic to $\z_p$.
Moreover with the help of the polynomial $\po h$
we define $\po b(x)=a'-\sum_{j\in \z_q}\po h(j\cdot x)$
to see that $\po b(c_0/\alpham)\ci a'/\beta$, while $\po b(c_j/\alpham)\ci e/\beta$ for $j\neq 0$.

Now, endowed with the polynomials $\po b, \oplus, +$ of $\m A$
and inspired by \cref{lm-zpqe} (with $m=q$)
we define an $s$-ary polynomial $\po p'$ of $\m A$ by putting
\[
\po p'(x_1,\ldots,x_s)=\sum_{(\o\beta,u)\in\m \z_m^s\times\z_m}
\mu_{\o\beta,u}\cdot \po b\left(\bigoplus_{i=1}^s \beta_i x_i\oplus u\right),
\]
with the coefficients $\mu_{\o\beta,u}\in \z_p$ chosen in a way that secures
\begin{equation}
\label{eq-spike}
\po p'(\po g(x_1),\ldots,\po g(x_s))\congruent{\beta}
\left\{
\begin{array}{ll}
a', &\mbox{if $f(\o x)=a$,}\\
e,  &\mbox{if $f(\o x)=e$,}
\end{array}
\right.
\end{equation}
whenever $\o x\in \set{c,d}^s$.
Obviously we cannot control the behavior of $\po p'$ over the entire algebra $\m A$,
but it suffices to have control over the trace $N$, or even over its subset $N'$.
Finally, congruence permutability applied to $(a,e)\in \alpham = \beta\join\Theta(a',e)$
gives $a''$ such that
$a\congruent{\beta}a''\congruent{\Theta(a',e)}e$.
In particular $\m A$ has a unary polynomial $\po g'$ with $\po g'\vpair{a'}{e} = \vpair{a''}{e}$.
Then the polynomial
$\po p(x_1,\ldots,x_s) = \po g'(\po p'(\po g(x_1),\ldots,\po g(x_s)))$
does the job, as $\po p(\o x)\congruent{\beta}f(\o x)$ for $\o x\in\set{c,d}^s$.
\end{proof}


\begin{proof}[Proof of \cref{lm:hard}]
Assume that $\m A$ is the smallest counterexample to our Lemma.
This means that two different primes $q_0,q_1$  occur in the interval $\intv{0}{\kappa}$.
Since this interval is a \pupi, 
this is witnessed by two subcovers, say $\gammaz,\gammaj$, of $\kappa$ with $\charr(\gammai,\kappa)=q_i$.
By minimality of our counterexample we know that $\gammaz\cap\gammaj=0$, as otherwise
-- in view of \cref{fact:quotient} --
the quotient algebra $\m A/(\gammaz\cap\gammaj)$ would be a smaller counterexample.
This, by modularity of $\con{A}$ gives that $\gammaz,\gammaj$ are atoms in $\con{A}$.
Again by minimality we argue that in fact  $\gammaz,\gammaj$ are the only atoms in $\con{A}$.
Indeed, if there is another atom $\theta$ in $\con{A}$
then in the smaller quotient $\m A/\theta$ of $\m A$ we recover our bad configuration,
namely the smallest co-supernilpotent congruence $\kappa_{\m A/\theta}=(\kappa\join\theta)/\theta$ with to its subcovers $(\gammai\join\theta)/\theta$ with $\charr((\gammai\join\theta)/\theta,(\kappa\join\theta)/\theta)=\charr(\gammai,\kappa)=q_i$.

Now, for $i=0,1$ we pick a maximal congruence $\fji$ that dominates $\gammai$ but not $\kappa$.
Obviously $\fji$ is meet-irreducible and we call its unique cover by $\fjpi$.
Since $\fji\not\geq\kappa\leq\fjpi$, we know that $\sr{\m A/\fjpi}=1$ while $\sr{\m A/\fji}=2$.
This means that the subdirectly irreducible algebra $\m A/\fji$ is not supernilpotent
so that in $\charrset{\m A/\fji}$ there is a prime different than $q_i$,
i.e. the characteristic of the monolith $\fjpi/\fji$ of $\m A/\fji$.
Again, $\sr{\m A/\fjpi}=1$ gives that the interval
$\intv{\fjpi}{1}$ is a \pupi and this different characteristics has to be witnessed by a
successor $\psii$ of $\fjpi$ with say $\charr(\fjpi,\psii)=p_i\neq q_i=\charr(\fji,\fjpi)$.

Now, for a minimal congruence $\alphai$ below $\psii$ but not below $\fjpi$,
we know that $\alphai$ is join-irreducible (with the unique subcover $\alphami$),
and that $\fjpi\prec\psii$ projects down to $\alphami\prec\alphai$.
In fact one can easily see that the entire interval $\intv{\fji}{\psii}$ project down to $\intv{\fji\cap\alphai}{\alphai}$ and we have:
\begin{senumerate}
\item
\label{fj1}
$\fji\cap\alphai \prec \alphami \prec \alphai$,
\item
\label{fj2}
$\charr(\fji\cap\alphai,\alphami)=q_i\neq p_i =\charr(\alphami,\alphai)$.
\end{senumerate}
In fact
\begin{senumerate}
\item
\label{fj3}
$\left(\fji\cap\alphai=0 \mbox{ and } \alphami=\gammao\right)$ \ \ or \ \
$\left(\fji\cap\alphai\geq\gammai \mbox{ and } \alphami\geq\kappa\right)$.
\end{senumerate}
Suppose first that $\fji\cap\alphai=0$.
Then $\alphami \succ \fji\cap\alphai=0$ is one of the atoms $\gammaz,\gammaj$.
But $\alphami=\gammai$ would give $q_{1-i}=\charr(0,\gammai)=\charr(\fji\cap\alphai,\alphami)=q_i$.
In the other case $\fji\cap\alphai$ is above one of the atoms $\gammaz,\gammaj$.
But $\fji\geq\fji\cap\alphai\geq\gammao$ together with $\fji\geq\gammai$ gives
$\fji\geq\gammai\join\gammao=\kappa$ contrary to our choice of $\fji$.
To see that in this case
$\kappa \leq \alphami$ suppose otherwise to get
$\kappa > \kappa\cap\alphami \geq \gamma_i$
and consequently $\gamma_i = \kappa\cap\alphami$.
Therefore $\intv{\gammai}{\alphai} \nearrow \intv{\kappa}{\kappa\join\alphai}$ so that
$p_i,q_i \in \charrset{\fji\cap\alphai,\alphai} \ci \charrset{\gammai,\alphai} = \charrset{\kappa,\kappa\join\alphai}$.
This however is not possible as $\alphai$ is join-irreducible in the lattice $\con{\m A}$
so that $\kappa\join\alphai$ is join-irreducible in the \pupi $\intv{\kappa}{1}$
and thus the interval $\intv{\kappa}{\alphai}$ has only one characteristic (see \cref{fact:ji-PUPI}).

\begin{senumerate}
\item
\label{fj4}
$\intv{0}{\gammao}\nearrow
\intv{\fji\cap\alphai}{\alphami}\nearrow
\intv{\fji}{\fjpi}$
\end{senumerate}
As $\intv{0}{\gammao}\nearrow \intv{\gammai}{\kappa}\nearrow
\intv{\fji}{\fjpi}$
then (\ref{section:hard}.\ref{fj4}) is shown for $\fji\cap\alphai=0$.
Otherwise (\ref{section:hard}.\ref{fj3}) gives $\fji\cap\alphai \geq \gammai$
and $\intv{\gammai}{\kappa}\nearrow\intv{\fji\cap\alphai}{\alphami}$, as required.

\medskip
Now note that (\ref{section:hard}.\ref{fj2}) gives
$p_i,q_i \in \charrset{\fji\cap\alphai,\alphai} \ci \charrset{\gammai,\alphai}$.
However we can show that
\begin{senumerate}
\item
\label{fj5}
$q_i\not\in\charrset{\gammai,\fji\cap\alphai}$
\end{senumerate}
Indeed, if $q_i\in\charrset{\gammai,\fji\cap\alphai}$
then in particular $\gammai < \fji\cap\alphai$
and by (\ref{section:hard}.\ref{fj3})
we have $\kappa \leq \alphami < \alphai$.
But then modularity gives
$(\fji\cap\alphai)\join\kappa=(\fji\join\kappa)\cap\alphai=\fjpi\cap\alphai=\alphami$,
while from $\intv{\gammai}{\kappa}\nearrow\intv{\fji}{\fjpi}$ we get
$(\fji\cap\alphai)\cap\kappa=\fji\cap\kappa=\gammai$.
These two equalities result in $\intv{\gammai}{\fji\cap\alphai} \nearrow \intv{\kappa}{\alphami}$,
which together with $p_i=\charr(\fjpi,\psii)=\charr(\alphami,\alphai)\in\charrset{\kappa,\alphai}$
and join-rreducibility of $\alphai$ in the \pupi $\intv{\kappa}{1}$ gives
$\charrset{\gammai,\fji\cap\alphai}=\charrset{\kappa,\alphami}=\set{p_i}\not\ni q_i$ as claimed.

\medskip
We fix an element $e\in A$ and use \cref{lm:minset} to get $(\fji\cap\alphami,\alphami)$-minimal set $\vi$ that contains $e$.
Note that the minimal sets $\vi$ have different characteristics, so that they cannot be equal.
Therefore by \cite[Lemma 4.30]{hm} we have
\begin{senumerate}
\item
\label{fj11}
$\vz\cap\vj=\set{e}$.
\end{senumerate}
Now let $N_i=\vi \cap e/\alphami$ is a $(\fji\cap\alphami,\alphami)$-trace in the minimal set $\vi$.
Obviously, due to (\ref{section:hard}.\ref{fj4}) the set $\vi$ is also $(0,\gammao)$-minimal,
but in general it may have smaller traces.
This however is not possible as, due to (\ref{section:hard}.\ref{fj5}), we have
\begin{senumerate}
\item
\label{fj12}
$\alphami|_{\vi} \ci \gammao$ and $\fji\cap\alphami|_{\vi} = 0$.
\end{senumerate}
Indeed, we fix a chain of congruences
$\gammao = \theta_0\prec\theta_1\prec\ldots\prec\theta_s=\alphami$
to show that $\theta_j|_{\vi}\ci\theta_{j-1}$.
First observe that $\intv{\theta_j\cap\fji\cap\alphai}{\theta_j}\searrow\intv{\theta_{j-1}\cap\fji\cap\alphai}{\theta_{j-1}}$
so that $\vi$ is $(\theta_j\cap\fji\cap\alphai,\theta_j)$-minimal.
If $\theta_j|_{\vi}\not\ci\theta_{j-1}$ then the set $\vi$,
as the range of some unary idempotent polynomial of $\m A$,
would contain $(\theta_{j-1},\theta_j)$-minimal set of characteristic different from $q_i$,
the characteristic of $\vi$.
Again this is not possible in view of \cite[Lemma 4.30]{hm},
and we are done with the first part of (\ref{section:hard}.\ref{fj12}).
To see the second one simply note that $\fji\cap\alphami|_{\vi}\ci\fji\cap\gammao=0$.

Now we fix $(\cci,\ddi)\in \alphai-\alphami$ and $\aai\in \vi\setm\set{\ee}$.
Note that, due to (\ref{section:hard}.\ref{fj12}) we have even $\aai\in \vi\setm\ee/\betai$,
where $\beta_i=\fji\cap\alphami$.

With the help of \cref{lm:beta-int} we are going to show the following claim:
\begin{senumerate}
\item
\label{gfq-pol}
For a polynomial $w(\o x)$ over the field $GF(q_i)$ of degree $s$
one can construct (in $n^{O(s)}$ steps) a polynomial $\po p_w(\o x)$ of $\m A$
such that for $\o x\in\set{\cci,\ddi}^n$ we have
\begin{equation*}
\po p_w(x_1,\ldots,x_n)=
\left\{
\begin{array}{ll}
\ee,   &\mbox{if $w(\pi_i(x_1),\ldots,\pi_i(x_n))=0$,}\\
\aai,  &\mbox{if $w(\pi_i(x_1),\ldots,\pi_i(x_n))=1$,}
\end{array}
\right.
\end{equation*}
where $\pi_i: A\map GF(q_i)$ satisfies $\pi_i\vpair{\cci}{\ddi} = \vpair{0}{1}$.
\end{senumerate}
We first prove (\ref{section:hard}.\ref{gfq-pol})
for monomials $u(x_1,\ldots,x_s)$ over $GF(q_i)$ of degree $s$.
\cref{lm:beta-int} allows us to $(\fji\cap\alphami)$-interpolate the monomial $u$ by a polynomial,
say $\po p_u(x_1,\ldots,x_s)$ of $\m A$ so that for $\o x\in\set{\cci,\ddi}^s$ we have
\begin{equation*}
\po p_u(x_1,\ldots,x_s) \congruent{\fji\cap\alphami}
\left\{
\begin{array}{ll}
\ee,   &\mbox{if $u(\pi_i(x_1),\ldots,\pi_i(x_s))=0$,}\\
\aai,  &\mbox{if $u(\pi_i(x_1),\ldots,\pi_i(x_s))=1$.}
\end{array}
\right.
\end{equation*}
Since $\set{\ee,\aai}\ci\vi$ we may additionally assume
that the range of $\po p_u$ is contained in $\vi$,
or simply replace $\po p_u$ by $\po e_{\vi}\po p_u$,
where $\po e_{\vi}$ is an idempotent polynomial of $\m A$ with the range $\vi$.
But then the second part of (\ref{section:hard}.\ref{fj12}) allows us to replace
$\congruent{\fji\cap\alphami}$ by the equality in the above display.

Now with the help of the binary polynomial $x+y=\po d(x,e,y)$ we sum up all the polynomials
$\po p_u$, with $u$ ranging over the monomials  $w$ to get $\po p_w\in \pol{A}$.
Note that each $\po p_u$ can be obtained in $2^{O(s)}$ steps, so that $\po p_w$ requires
$n^{O(s)}$ steps.
Note also that for $(x_1,\ldots,x_s)\in\set{\cci,\ddi}^s$
the values $\po p_u(x_1,\ldots,x_s)$ are in $\set{\ee,\aai}$,
so that their sum $\po p_w(\o x)$ lies in the cyclic subgroup
$\langle \aai \rangle$ of $(\vi\cap \ee/\alphami; + ; \ee)$ generated by $\aai$.
Summing the polynomials $\po p_u$ in the group $\langle \aai \rangle$ isomorphic to $\z_{q_i}$
corresponds to summarizing monomials of $w$ in the underlying additive group $\z_{q_i}$ of $GF(q_i)$.
Thus the values $\po p_w(\o x)$ correspond to the values $w(\pi_i(x_1),\ldots,\pi_i(x_n))$
in the way described in (\ref{section:hard}.\ref{gfq-pol}).

\medskip
To conclude our proof we associate with each 3-CNF formula
$\Phi$ (with $\ell$ clauses and $n$ variables $x_1,\ldots,x_n$)
a program $\progg{\po p^\Phi}{n}{\iota}{\ee}$ over $\m A$ such that $\o b\in\bool^n$ satisfies $\Phi$ iff
$\prog{\po p^\Phi}{\iota}(\o b)=\ee$.

The polynomial $\po p^\Phi(x_1^0,\ldots,x_n^0,x_1^1,\ldots,x_n^1)$ is going to be $2n$-ary.
The instructions of our program are given by
\(
\iota(x_j^i)=(\b_j, c_i, d_i).
\)
To construct the polynomial $\po p^\Phi$ we first pick two integers $\nu_0,\nu_1$
satisfying
\(
q_i^{\nu_i-1}\leq\sqrt{\ell}<q_i^{\nu_i}.
\)
Then, with $i=0,1$, we use \cref{lm:pseudo-and} to get the polynomials
$w^\Phi_i(\o x)$ over $GF(q_i)$ of degree $O(q_i^{\nu_i})$
so that for $\o b\in\bool^n$ we have
\[
w^\Phi_{i}(\o b) =
\left\{
\begin{array}{ll}
0, &\mbox{if the number of unsatisfied (by $\o b$) clauses in $\Phi$ is divisible by $p_i^{\nu_i}$}\\
1, &\mbox{otherwise.}
\end{array}
\right.
\]
Now we use (\ref{section:hard}.\ref{gfq-pol}) to translate the $w^\Phi_{i}$'s
into the polynomials $\po p^\Phi_{i}$ of $\m A$
that satisfy for $\o x\in\set{\cci,\ddi}^n$
\[
\po p^\Phi_{i}(\o x) =
\left\{
\begin{array}{ll}
0, &\mbox{if the number of unsatisfied (by $\pi_i(\o x)$) clauses in $\Phi$ is divisible by $p_i^{\nu_i}$}\\
a_i, &\mbox{otherwise.}
\end{array}
\right.
\]
Finally we let $\po p^\Phi$ to be the $2n$-ary polynomial of $\m A$ defined by
\(
\po p^\Phi(\o x^0,\o x^1)=\po d(\po p^\Phi_{0}(\o x^0),\po p^\Phi_{1}(\o x^1),\ee),
\)
with $\o x^0,\o x^1$ being disjoint sets of variables.
Obviously for $\po p^\Phi_{0}(\o x^0)=\ee$ and $\po p^\Phi_{1}(\o x^1)=\ee$
we have $\po p^\Phi(\o x^0,\o x^1)=\ee$.
But also $\po p^\Phi(\o x^0,\o x^1)=\ee$ implies $\po p^\Phi_{0}(\o x^0)=\po p^\Phi_{1}(\o x^1)$,
by the fact that Malcev polynomial in nilpotent algebra is a permutation
with respect to the first variable (see \cite[Lemma 7.3]{fm}) and
$\po d(\po p^\Phi_{1}(\o x^0),\po p^\Phi_{1}(\o x^1),\ee)=\ee$.
But then
\(
V_0\ni \po p^\Phi_{0}(\o x^0)=\po p^\Phi_{1}(\o x^1)\ni V_1
\)
together with (\ref{section:hard}.\ref{fj11}) gives
$\po p^\Phi_{0}(\o x^0)=\ee=\po p^\Phi_{1}(\o x^1)$.
Thus
\(
\ee = \prog{f}{\iota}(\o b) = \po p^\Phi(\pi_0^{-1}(\o b),\pi_1^{-1}(\o b))
\)
is equivalent to
$\po p^\Phi_{0}(\pi_0^{-1}(\o b))=\ee=\po p^\Phi_{1}(\pi_1^{-1}(\o b))$
or, in other words tells us that the number of unsatisfied (by $\o b$) clauses in $\Phi$
is divisible by both $p_0^{\nu_0}$ and $p_1^{\nu_1}$.
But this together with $\sqrt{\ell}<q_i^{\nu_i}$
tells us that $\o b$ satisfies the formula $\Phi$.

To complete our proof we bound the time needed to construct the polynomial $\po p^\Phi$.
First note that by \cref{lm:pseudo-and} the polynomials $w^\Phi_i$ of $GF(q_i)$
can be obtained in $2^{O(q_i^{\nu_i}(\log n+\log p))}=2^{O(\sqrt{\ell}\log n)}$ steps.
Then (\ref{section:hard}.\ref{gfq-pol}) helps us to translate the $w^\Phi_i$'s
into the polynomials $\po p_i^\Phi$ of $\m A$ in $n^{O(\sqrt(l))}=2^{O(\sqrt{l}\log n)}$ steps.
Thus also the final polynomial $\po p^\Phi$ is computable in $2^{O(\sqrt{l}\log n)}$ steps.
But this bound, together with \rethh,
obviously blocks the existence of a polynomial time probabilistic algorithm solving \progcsat{\m A}. 
\end{proof}


\section{ProgramCSat -- easiness}
\label{section:easy}


We start this chapter by presenting number of usefull facts about low-depth circuits  which use only arithmetic operations and conjunction. Similar results can be found for instance in \cite{Grolmusz01, GrolmuszT00}.

\begin{lm}
\label{lm:and-sum}
Every function of the form $\set{0,1}^n \map \z_p^k$
can be computed by an $\ccand_n\circ\sumpk{p}{k}$-circuit of size at most $2^n$.
\end{lm}

\begin{proof}
Our function can be extended to $GF(p^k)\map GF(p^k)$
and then represented by a polynomial of $GF(p^k)$ of degree $n\cdot(p^k-1)$.
However, due to the fact that for $x\in\set{0,1}$ and $j\geq1$ we have $x^j=x$, each nonconstant monomial can be replaced by the one of degree between $1$ and $n$.
Such monomials, with unit coefficients, can obviously be represented by $\ccand_n$-gates.
The other coefficients as well as the sums of monomials with arbitrary coefficients are covered by the $\sumpk{p}{k}$-gate, since multiplication by element of $GF(p^k)$ is a linear map when we treat $GF(p^k)$ as a $k$ dimensional vector space over $GF(p)$.
\end{proof}

\begin{lm}
\label{lm:normal-form}
Let $m$ be a square-free positive integer, $p\nmid m$ be a prime.
Then every function computable by a $\ccmod_m\circ\ccand_d$-circuit
can be also computed by a $\ccmod_m\circ\sumpk{p}{1}$-circuit of size $O(m^{d+1})$.
\end{lm}
\begin{proof}
The boolean function $f$ computable by a $\ccmod_m\circ\ccand_d$-circuit can be decomposed into
$f(\o x) = \ccand_d(\chi_{S_1}g_1(\o x),\ldots,\chi_{S_d}g_d(\o x))$
for some $\z_m$-linear functions $g_1,\ldots,g_m$ and characteristic functions $\chi_{S_i}$
of some subsets $S_1,\ldots,S_d \ci \z_m$.
With the help of \cref{lm-zpqe} the $d$-ary function
\[
\z_m^d \ni (y_1,\ldots,y_d) \mapsto
\ccand_d(\chi_{S_1}(y_1),\ldots,\chi_{S_d}(y_d))\in \set{0,1}\ci\z_p
\]
can be presented in the form
\(
\sum_{(\o\beta,u)\in\m \z_m^d\times\z_m}
\mu_{\o\beta,u}\cdot \b\left(\bigoplus_{i=1}^d \beta_i y_i\oplus u\right).
\)
Substituting $y_j$ by $g_j(\o x)$ we end up with a function of the same form
in which the inner sum $\bigoplus$ increases in arity,
but the outer sum has the same number of summands bounded by $m^{d+1}$.
Obviously the expressions of the form $\b\left(\bigoplus \ldots\right)$
can be computed by $\ccmod_m$-gates,
while the outer sum is $\sumpk{p}{1}$-computable.
\end{proof}

\begin{lm}
\label{lm:unmod}
Every function computable by a $\ccmod_m\circ\ccmod_p$-circuit of size $\csize$
can be also computed by a $\ccmod_m\circ\sumpk{p}{1}$-circuit of size $O(m^p\csize^p)$.
\end{lm}
\begin{proof}
This time a $\ccmod_m\circ\ccmod_p$-computable function $f(\o x)$ can be decomposed into\break
\(
\chi_T\left(\sum_{i=1}^{\csize-1} \alpha_i\cdot\chi_{S_i}g_i(\o x)\right)
\)
where $T\ci\z_p$, $S_i\ci\z_m$ and the $g_i$'s are $\z_m$-linear functions.
Note that $\chi_T(z)$ can be represented by a unary polynomial, of degree $p-1$ over $GF(p)$.
Thus
\(
\chi_T\left(\sum_{i=1}^{\csize-1} \alpha_i y_i\right)
\)
is also a polynomial over $GF(p)$ of $\csize-1$ variables with degree $p-1$.
Thus this polynomial, when restricted to $\set{0,1}$, is computable by
$\ccand_{p-1}\circ\sumpk{p}{1}$-circuit of size $\csize^{p-1}$,
i.e. the maximal possible number of monomials here.
Since the $\chi_{S_i}g_i(\o x)$'s are $\ccmod_m$-computable
we end up with $\ccmod_m\circ\ccand_{p-1}\circ\sumpk{p}{1}$-circuit.
However \cref{lm:normal-form} help us to replace $\ccmod_m\circ\ccand_{p-1}$-parts
by $\ccmod_m\circ\sumpk{p}{1}$-circuits of size $O(m^p)$
so that the size of the final circuit computing $f$ can be bounded by $O(m^p\csize^p)$.
\end{proof}

\begin{lm}
\label{lm:apply_func}
Let $m$ be a square-free positive integer, $p\nmid m$ be a prime 
and $g$ be a $k$-ary boolean function.
If the functions $f_1, \ldots, f_k$ are computable by
$\ccmod_m\circ\ccmod_p$-circuits
of size $O(\lambda)$
then also $g(f_1,\ldots,f_k)$ is computable by such a circuit,
but of size $O(\lambda^{kp})$.
\end{lm}
\begin{proof}
As in the proof of \cref{lm:and-sum} the function $g$ can be represented as a restriction of a polynomial $g'(z_1,\ldots,z_k)$ of degree $k$ over $GF(p)$ to the set $\set{0,1}^k$.
On the other hand, by \cref{lm:unmod}, each of the $f_j$'s is computable by a $\ccmod_m\circ\sumpk{p}{1}$-circuit of size $O(m^p\csize^p)$.
In particular $f_j(\o x)$ can be decomposed into
\(
\sum_{i\in I_j} \alpha_i\cdot\chi_{S_i^j}g_i^j(\o x)
\)
where $S_i^j\ci\z_m$ and the $g_i^j$'s being $\z_m$-linear.
Now plugging the linear polynomials
\(
f'_j(\ldots,y_i^j,\ldots)=\sum_{i\in I_j} \alpha_i\cdot y_i^j
\)
over $GF(p)$, with at most $O(m^p\csize^p)$ variables $y_i^j$, into $g'(z_1,\ldots,z_k)$
we get a polynomial of degree $k$ which, as before,
on the arguments from $\set{0,1}$ is $\ccand_k\circ\sumpk{p}{1}$-computable.
Note that the polynomial $g'(f'_1,\ldots,f'_k)$ may have up to $O(km^p\csize^p)$ variables, as the $f_j$'s do not need to share variables.
Thus the size of $g'(f'_1,\ldots,f'_k)$ and therefore of the corresponding  $\ccand_k\circ\sumpk{p}{1}$-circuit is bounded by the number of monomials of degree $k$ over $O(km^p\csize^p)$ variables, i.e. roughly by $O((km^p\csize^p)^k)$.
To compute $g(f_1,\ldots,f_k)$ we simply replace the $y_i^j$'s in the above circuit computing
$g'(f'_1,\ldots,f'_k)$ by $\ccmod_m$-circuits computing $\chi_{S_i^j}g_i^j(\o x)$
to get $\ccmod_m\circ\ccand_k\circ\sumpk{p}{1}$-circuit of size $O((km^p\csize^p)^k)$.
However, in view of \cref{lm:normal-form}, this can be replaced by $\ccmod_m\circ\sumpk{p}{1}$-circuit, by the expense of increasing the size $O(m^{k+1})$ times.
Finally we note that since $g$ returns only $0$ or $1$ the $\sumpk{p}{1}$-gate can be replaced by
$\ccmod_p$ type gate. In our applications we treat $m,k$ and $p$ as constants, so the size of the resulting circuit is $O(\lambda^{kp})$.
\end{proof}

\begin{lm}
\label{lm:5to3}
Let $m$ be a square-free positive integer, $p\nmid m$ be a prime and $\ccc\in\z_p^\nu$.
Then every $\ccand_d\circ\ccmod_m\circ\ccmod_p
\circ\ccand_{d'}\circ\sumpk{p}{\nu}^{\ccc}$-circuit of size $\lambda$
can be replaced a
$\ccand_d\circ\ccmod_m\circ\ccmod_p$-circuit
of size $O(\lambda^{\nu d'p^3})$.
\end{lm}
\begin{proof}
We start with noting that for a constant $c\in \z_p^\nu$
an equation of the form $t=c$ in the module $\z_p^\nu$
can be replaced by a system of $\nu$ equations $t^{[j]}=c^{[j]}$
(with $t^{[j]}$ being the projection of $t$ onto the $j$-th coordinate)
and then by a single equation $g(t^{[1]},\ldots,t^{[\nu]})=1$, where
\(
g(y_1,\ldots,y_\nu)= \prod_{j=1}^\nu (1-(y_j-c^{[j]})^{p-1})
\)
is a polynomial of degree $\nu\cdot (p-1)$ over $GF(p)$.

Now we observe that our starting function can be actually computed by a
$\ccand_d\circ\ccmod_m\circ\ccmod_p\circ\sumpk{p}{\nu}^{\ccc}$-circuit, by simply
using \cref{lm:apply_func} to eliminate the $\ccand_{d'}$-level.
Thus this function can be represented (with a sum in the group $\z_p^\nu$) as
$\chi_\set{c}(f(\o x))$ for
\(
f(\o x) = \sum_{i=1}^{s} \alpha_i\cdot f_i(\o x)
\)
with $\alpha_i$ being the $k\times k$ matrices over the field $GF(p)$
and $f_1,\ldots,f_s$ being $\ccand_d\circ\ccmod_m\circ\ccmod_p$-computable.
Formally the values $f_i(\o x)$ are boolean,
but we interpret them by the constant $\nu$-tuples $(f_i(\o x),\dots,f_i(\o x))\in\z_p^\nu$.

Observe also that $f(\o x)^{[j]}=\sum_{i=1}^{s} \alpha_{ij}\cdot f_i(\o x)$ for $\alpha_{ij}\in\z_p$ is a sum of the $j$-th row of the matrix $\alpha_{ij}$,
where this time the boolean values taken by $f_i$ are interpreted by $0,1\in\z_p$.
The polynomial $g'$ obtained from $g$ by
\(
g'(z_1,\ldots,z_s)=g(\sum_{i=1}^{s} \alpha_{i1}\cdot z_i,\ldots,\sum_{i=1}^{s}\alpha_{i\nu}\cdot z_i)
\)
has also degree bounded by $\nu p$ so that, on the set $\set{0,1}$,
it is $\ccand_{\nu p}\circ\sumpk{p}{1}$-computable.
Thus checking if $f(\o x)=\ccc$ reduces to checking $g'(f_1(\o x),\ldots,f_s(\o x))=1$
which in turn can be done with the help of $\ccand_d\circ\ccmod_m\circ\ccmod_p
\circ\ccand_{\nu p}\circ\sumpk{p}{1}$-circuit.
Now \cref{lm:apply_func} and \cref{lm:unmod} allows us to eliminate the $\ccand_{\nu p}$-level,
then replace the two levels $\ccmod_m\circ\ccmod_p$ by $\ccmod_m\circ\sumpk{p}{1}$,
and finally compress two $\sumpk{p}{1}$-levels to a single one and replace it by $\ccmod_p$.

A careful analysis of this process gives the required bound for the final
$\ccand_d\circ\ccmod_m\circ\ccmod_p$-circuit.
\end{proof}


Now we present a characterization of Boolean functions that can be represented in supernilpotent algebras.

\begin{lm}
\label{lm:supernil-circuit}
Let $\m A$ be a finite supernilpotent Malcev algebra.
Then the functions computable by an $n$-ary boolean program
$\progg{\po p}{n}{\iota}{S}$ over the algebra $\m A$
can be also computed by
$\ccand_d\circ\ccmod_{\pdiv A}\circ\ccor_{\card S}$-circuits of size $O(n^d)$ with
$d\leq\card{A}^{1+\log\maxar A} = (2 \cdot \maxar A)^{\log\card{A}}$.
\end{lm}

\begin{proof}
First we reduce our setting to $\card S =1$ as, after building appropriate $\ccand_d\circ\ccmod_{\pdiv A}$-circuits separately for each $\ccc \in S$,
we eventually join them by the final gate $\ccor_{\card S}$.

Next we note that the algebra $\m A$, being supernilpotent,
decomposes into a product $\m A_1 \times\ldots\times \m A_\s$ of supernilpotent algebras $\m A_i$,
each of which is of prime power order, say $\card{A_i}=p_i^{\nu_i}$,
where $\sdiv A=\set{p_1,\ldots,p_\s}$.

In order to construct $\ccand_d\circ\ccmod_{\pdiv A}$-circuit computing the boolean function
$\progb{\po p}{\iota}{\ccc}$ we first produce an $n$-ary polynomial $w(\b_1,\ldots,\b_n)$
of degree $d\leq\card{A}^{1+\log\maxar A}$ over the ring $\z_{\pdiv A}$
such that for all $\o b\in \bool$ we have
\[
\progb{\po p}{\iota}{\ccc}(\o b)=\true \mbox{\quad iff \quad}
w(\o b)=\ppdiv',
\]
where $\ppdiv'=\frac{\pdiv A}{p_1}+\ldots+\frac{\pdiv A}{p_\s}$.
Given the polynomial $w(\b_1,\ldots,\b_n)$ as above
we translate it into $\ccand_d\circ\ccmod_{\pdiv A}$-circuit
by imitating the monomials of $w$ (of degree at most $d$) by $\ccand_d$-gates
and then summing them up by an appropriate
$\ccmod_{\pdiv A}$-gate with accepting set $\set{\ppdiv'}$.
Obviously the bound for the degree of the polynomial $w$ determines the upper bound for the size of just constructed circuit, as there are at most $O(n^d)$ monomials of degree $d$.

The polynomial $w$ is going to be obtained separately for each factor in our decomposition
$\m A=\m A_1 \times\ldots\times \m A_\s$.
In what follows an element $a$ of $A$ is to be presented as a tuple $(a(1),\ldots,a(\s))$.
Moreover our decomposition of $\m A$ leads to the decomposition of the program
$\progg{\po p}{n}{\iota}{\ccc}$ over $\m A$ into $\s$ programs
$\progg{\po p_j}{n}{\iota_j}{\ccc_j}$, over $\m A_j$ respectively,
where
\begin{itemize}
  \item $\po p_j$ is the polynomial of $\m A_j$ corresponding to $\po p$,
  \item $\iota_j(x_i)=(\b^x,a^x_0(j),a^x_1(j))$ whenever $\iota(x)=(\b^x,a^x_0, a^x_1)$,
  \item $\ccc_j=\ccc(j)$.
\end{itemize}
Note here that for this decomposition we have
\begin{align*}
\progb{\po p}{\iota}{\ccc}(\o b)=\true
&\mbox{\quad iff \quad}
\po p(a^{x_1}_{b^{x_1}},\ldots,a^{x_k}_{b^{x_k}})=\prog{\po p}{\iota}(\o b)=c
\\
&\mbox{\quad iff \quad}
\amper_{j=1}^{\s} \ \po p(a^{x_1}_{b^{x_1}}(j),\ldots,a^{x_k}_{b^{x_k}}(j))
=\prog{\po p_j}{\iota_j}(\o b)=c(j)
\\
&\mbox{\quad iff \quad}
\amper_{j=1}^{\s} \ \progb{\po p_j}{\iota_j}{\ccc_j}(\o b)=\true
\end{align*}
The plan is to construct, for each $j=1,\ldots,\s$ an appropriate polynomial $w_j(\b_1,\ldots,\b_n)$
(of degree bounded by $\card{A}^{1+\log\maxar A}$)
over the prime field $GF(p_j)$
satisfying
\[
\progb{\po p_j}{\iota_j}{\ccc_j}(\o b)=\true \mbox{\quad iff \quad} w_j(\o b)=\true
\]
and then sum them up to get the polynomial
$w=\frac{\pdiv A}{p_1}\cdot w_1+\ldots+\frac{\pdiv A}{p_\s}\cdot w_\s$
over the ring $\z_{\pdiv{A}}$
that would clearly satisfy
\[
w(\o b) = \ppdiv
\mbox{\quad iff \quad}
\amper_{j=1}^{\s} \ w_j(\o b)=\true
\mbox{\quad iff \quad}
\progb{\po p}{\iota}{\ccc}(\o b)=\true,
\]
as promised.

Before we construct the $w_j$'s note that despite the fact that their variables $\b_1,\ldots,\b_n$ are denoted here to be boolean, they in fact may range over entire $GF(p_j)$
(and finally over entire $\z_{\pdiv{A}}$) but we do not care what they return on nonboolean arguments from the set $\bool$.

\medskip

Now we are left with a supernilpotent algebra $\m A$ of size $p^\nu$ with $p$ being a prime.
In such algebra all prime quotients of congruences have characteristic $p$ (since $\m A$ has classes of congruences of equal size).
Thus, by \cref{lm:simple-module-atom} we know that for each $\zero\in A$
and each prime quotient $\alpha\prec\beta$
the binary operation $\po d(x,\zero,y)$ defines (modulo $\alpha$) the structure of an elementary $p$-group on the set $\set{a/\alpha : a\in\zero/\beta}$, i.e. a group isomorphic to some power of the group $\z_p$.
Thus, going up along the chain
$0=\beta_0 \prec \beta_1 \prec \ldots \beta_{t-1} \prec \beta_t=1$
we can recursively conclude that for each  $j$ the coset $\zero/\beta_j$ can be endowed with a group structure isomorphic to some power of $\z_p$.
Indeed, if $\zero/\beta_j=a_1/\beta_{j-1}\cup\ldots a_{p^{\nu_j}}/\beta_{j-1}$
and each coset $a_i/\beta_{j-1}$ can be endowed with a group structure isomorphic to
$\z_p^{\nu_1+\cdots+\nu_{j-1}}$ then the entire coset can be trated as having the (external) structure of the group product $\z_p^{\nu_1+\cdots+\nu_{j-1}}\times\z_p^{\nu_j}$.
This global group operation (which we call external)
is not necessarily determined by a polynomial of $\m A$
but it does not change the congruences of $\m A$.
The reader may want to consult \cite{aichinger-spectrum, KawalekK}
where the structure of supernilpotent algebras of prime power order has been studied in more details.

In particular \cite{KawalekK} ensures us that the decomposition of the set $A$ into the product $Z_p^\nu$ translates a little bit further.
Before stating this correspondence we present each element $a\in A$ as a tuple
$(a^{(1)},\ldots a^{(\nu)}) \in Z_p^\nu$.
Next we treat $Z_p$ not only as the underlying set of the group $(\z_p;+)$
but also as the underlying set of the prime field $GF(p)$.
Now, after this identification of $A$ with $Z_p^\nu$,
\cite[Lemma 3.2]{KawalekK} allows to simulate $k$-ary polynomials of the algebra $\m A$
by the $k\nu$-ary polynomials of the field $GF(p)$.
In particular for the $k$-ary polynomial $\po p(x_1,\ldots,x_k)$
from our program $\progg{\po p}{n}{\iota}{\set{c}}$ there is a polynomial
$w'(x_1^{(1)},\ldots,x_1^{(\nu)},\ldots,x_k^{(1)},\ldots,x_k^{(\nu)})$
over the field $GF(p)$ with degree bounded by $\card{A}^{1+\log_p\maxar A}$,
such that for all $(a_1,\ldots,a_k)\in A^k$ we have
\[
\po p(a_1,\ldots,a_k)=c
\mbox{\quad iff \quad }
w'(a_1^{(1)},\ldots,a_1^{(\nu)},\ldots,a_k^{(1)},\ldots,a_k^{(\nu)})=1.
\]
Now to construct an $\ccand_d\circ\ccmod_{\pdiv A}$-circuit computing
the function
$\progb{\po p}{\iota}{\ccc}(b_1,\ldots,b_n):\bool^n \map \bool$
we first modify the $k\nu$-ary polynomial
$w'$ to $n$-ary polynomial $w(\b_1,\ldots,\b_n)$,
again over the field $GF(p)$ and again of degree $d \leq \card{A}^{1+\log_p\maxar A}$.
We simply replace the variables $x_i^{(1)},\ldots,x_i^{(\nu)}$,
by affine transformations of the corresponding boolean variable $\b^{x_i}$
(determined by the program instruction $\iota(x_i)=(\b^{x_i},a_0^{x_i}, a_1^{x_i})$
for the variable $x_i$ of $\po p$) as follows:

\begin{align*}
x_i^{(1)} &\mbox{\quad by}
&\b^{x_i} &\cdot \left((a_1^{x_i})^{(1)}-(a_0^{x_i})^{(1)}\right)+(a_0^{x_i})^{(1)},
\\
&&\vdots
\\
x_i^{(\nu)} &\mbox{\quad by}
&\b^{x_i} &\cdot \left((a_1^{x_i})^{(\nu)}-(a_0^{x_i})^{(\nu)}\right)+(a_0^{x_i})^{(\nu)}.
\end{align*}
It should be easy to observe that then for $\o b=(b_1,\ldots,b_n)\in \bool^n$ we have
\[
\progb{\po p}{\iota}{\ccc}(\o b)=\true
\mbox{\quad  (i.e. $\prog{\po p}{\iota}(\o b)=\ccc$) \quad iff \quad}
w(\o b)=\true
\]
so that we are done.
\end{proof}

\begin{thm}[\cref{thm:2supernil-circuit-early} restated]
\label{thm:2supernil-circuit}
Let $\m A$ be a finite nilpotent Malcev algebra with $\charrset{0,\kappa}=\set{p}$.
Then the function computable by a boolean program
$\progg{\po p}{n}{\iota}{S}$ of size $\ell$ over the algebra $\m A$
can be also computed by an
$\ccand_d\circ\ccmod_{\pdiv A/p}\circ\ccmod_p$-circuits of size $O(\ell^c)$
with $d\leq\card{A}^{1+\log\maxar A}$
and the degree $c$  bounded by $O(\card{A}\cdot\maxar{A})^{O(\log\card{A})}$.
\end{thm}

\begin{proof}
As in the proof of \cref{lm:supernil-circuit} we additionally assume that $\card S =1$.
After producing $\ccand_d\circ\ccmod_{\pdiv A/p}\circ\ccmod_p$-circuits
separately for each $\ccc \in S$ we join them by $\ccor_{\card S}$-gate
to get $\ccand_d\circ\ccmod_{\pdiv A/p}\circ\ccmod_p\circ\ccor_{\card S}$-circuit.
But now \cref{lm:apply_func} allows us to completely eliminate $\ccor_{\card S}$-gate
by increasing the size of the resulting $\ccand_d\circ\ccmod_{\pdiv A/p}\circ\ccmod_p$-circuit
from $\csize$ to $\csize^{p\cdot|A|}$.
Thus, from now on we assume that $S=\set{\ccc}$.

The first of the following two claims shows that $\m A/\sigma_p$ is supernipotent,
while the second one that $\pdiv{A/\sigma_p}=\pdiv{A}\!/p$.
\begin{senumerate}
\item
\label{sig1}
$\kappa\leq\sigma_p\leq\sigma$,
\item
\label{sig2}
$p\not\in\charrset{\sigma_p,1}$.
\end{senumerate}
The second inequality in (\ref{section:easy}.\ref{sig1}) is obvious.
For the first one note that our assumption $\charrset{0,\kappa}=\set{p}$ tells that there is no other prime than $p$ below $\kappa$. Since $\sigma_p$ is the largest congruence with this property
it must be over $\kappa$.
To see (\ref{section:easy}.\ref{sig2}) first note that the congruence lattice of the supernilpotent algebra $\m A/\sigma_p$ is a \pupi,
so that each characteristic occuring in $\m A/\sigma_p$
has to occur in $\con{\m A}$ between $\sigma_p$ and one of its successors.
But $\sigma_p$ has no successors of characteristic $p$.

\medskip

Now \cref{lm:supernil-circuit} supplies us with a
$\ccand_d\circ\ccmod_{\pdiv{\m A/\sigma_p}}$-circuit
computing the boolean function
of the quotient program $\progg{\po p/\sigma_p}{n}{\iota/\sigma_p}{\ccc/\sigma_p}$.

The essential part of the proof is to go down along the chain of congruences
$\sigma_p=\gamma_h\succ\gamma_{h-1}\succ\ldots\succ\gamma_1\succ\gamma_0=0$
to consecutively produce $\ccand_d\circ\ccmod_{\pdiv{\m A/\sigma_p}}\circ\ccmod_p$-circuit
computing the boolean function
of the quotient program $\progg{\po p/\gamma_j}{n}{\iota/\gamma_j}{\ccc/\gamma_j}$.
This in turn will be possible after showing the following claim:
\begin{senumerate}
\item
\label{mod-beta}
Suppose $m$ is a square-free positive integer and $p\nmid m$ is a prime.
Let $\beta$ be an abelian congruence of characteristic $p$
in a finite nilpotent Malcev algebra $\m D$.
Moreover let $\progg{\po p}{n}{\iota}{\cc}$ be a boolean program of length $\ell$ over $\m D$.
If all programs of the form $\progg{\po p'/\beta}{n}{\iota/\beta}{\cc'/\beta}$
over the quotient algebra $\m D/\beta$,
(with $\po p'$ ranging over subterms of $\po p$ and $\cc'$ ranging over elements of $D$)
are computable by $\ccand_d\circ\ccmod_m\circ\ccmod_p$-circuits of size at most $\csize$
then $\progg{\po p}{n}{\iota}{\cc}$ is computable by a $\ccand_d\circ\ccmod_m\circ\ccmod_p$-circuit
of size $O((\csize\cdot\ell)^{2p^4\card{D}\maxar{D}})$
\end{senumerate}

With claim (\ref{section:easy}.\ref{mod-beta}) we are ready to bound the size of the final circuit.
We start with $\ccand_d\circ\ccmod_{\pdiv{\m A/\sigma_p}}$-circuit
computing the quotient program $\progg{\po p/\sigma_p}{n}{\iota/\sigma_p}{\cc}$ and
supplied by \cref{lm:supernil-circuit}.
Its size is bounded by $\csize_0\leq O(\ell^d)$ with $d\leq(2 \cdot \maxar A)^{\log\card{A}}$.
Now we go down along the chain of congruences
$\sigma_p=\gamma_h\succ\gamma_{h-1}\succ\ldots\succ\gamma_1\succ\gamma_0=0$
and use (\ref{section:easy}.\ref{mod-beta}) to get the circuit computing
$\progg{\po p}{n}{\iota}{\cc}$.
This requires $h\leq\card{A}$ iterations so that we end up with a circuit of size
$O((\ldots((\csize_0\cdot\ell)^\aexp\cdot\ell^\aexp)\ldots\cdot\ell)^\aexp)
\leq(\csize_0\cdot\ell)^{\aexp^{h+1}})$
where $\aexp\leq 2p^4\card{A}\maxar{A}\leq\card{A}^6\maxar{A}$.
One can combine all these inequalities to get the bound  $O(\ell^c)$
with $c\leq 2\card{A}\maxar{A}^{13\log\card{A}}$.
To go from circuit for $\progg{\po p}{n}{\iota}{\cc}$
to the one for $\progg{\po p}{n}{\iota}{S}$, with arbitrary $S\ci A$,
\cref{lm:apply_func} increases the size to $O((\ell^c)^{\card{A}^2})$,
but the exponent of $\ell$ is still bounded by
$O(\card{A}\cdot\maxar{A})^{O(\log\card{A})}$, as required.

\medskip

Now we return to the proof of claim (\ref{section:easy}.\ref{mod-beta})
and show how to construct the circuit for $\progg{\po p}{n}{\iota}{\cc}$
from the circuits over $\m D/\beta$ for the subterms of $\po p$.
To start this construction we note that (due to \cite[Exercise 7.3]{fm})
the atomic congruence $\beta$ lies below the center of $\m D$ so that,
following the proof of \cite[Theorem 7.1]{fm}, we can represent algebra $\m D$
in the following way:
\begin{itemize}
  \item the universe of $\m D$ is a direct product $M\times D'$, where
  \begin{itemize}
    \item $D'=D/\beta$, and
    \item $M$ is a $\beta$-coset of $\m D$
        (and therefore, by \cref{lm:simple-module-atom}, it is polynomially equivalent
        to a simple module over a subring $R$ of the ring $\matr{p}{\nu}$ of all endomorphisms of the abelian group $\z_p^\nu$),
  \end{itemize}
  \item each basic operation $f$ of $\m D$ is determined by
    a sequence of scalars $\alpha^f_1,\ldots,\alpha^f_{\ar f} \in R$
    and a function $\h f: (D')^{\maxar f} \map M$ in the following way
\begin{equation}\label{rec:basic}
    f^{\m D}(d_1,\ldots,d_{\ar f}) =
    \left(
    \sum_{i=1}^{\ar f} \alpha^f_i\cdot d^M_i +\h f(d_1^{D'},\ldots,d_{\ar f}^{D'}),
    f^{\m D'}(d_1^{D'},\ldots,d_{\ar f}^{D'})
    \right)
\end{equation}
    whenever $d_i=(d_i^M,d_i^{D'})\in M\times D'$.
\end{itemize}
Actually for a more complex polynomial $\po p(\o x) = f(\po p_1(\o x),\ldots,\po p_{\ar f}(\o x))$,
where $f$ is a basic operation and $\po p_1,\ldots,\po p_{\ar f}$ are polynomials the above representation still holds, i.e.,
\[
\po p^{\m D}(\o d) =
\left(
    \sum_{i=1}^{\ar{\po p}} \alpha^{\po p}_i\cdot d^M_i +
    \h {\po p}(d_1^{D'},\ldots,d_{\ar{\po p}}^{D'}),
    \po p^{\m D'}(d_1^{D'},\ldots,d_{\ar{\po p}}^{D'})
\right).
\]
However to understand the function $\h{\po p} : (D')^{\ar{\po p}} \map M$,
given for $\o{d'} \in (D')^{\ar{\po p}}$ by
$\h{\po p}(\o{d'}) = \sum_{i=1}^{\ar f} \alpha^f_i\cdot\h{\po p_i}(\o{d'})+
\h{f}(\po p_1^{\m D'}(\o{d'}),\dots,\po p_{\ar{f}}^{\m D'}(\o{d'}))$
we need to recursively understand the $\h{\po p_i}$'s.
This in turn requires to go deeper into the structure of how the polynomial $\po p$
[or corresponding circuit] is build from the basic operations [or nodes, respectively].
Thus for a basic operation $f$ [or a node in a circuit labeled by $f$]
we form the set $N_f$ of all occurrences of this symbol [node].
Each such occurrence of $f$, say $f^{(j)}$ with $j\in N_f$ has its inputs determined by the polynomials/nodes $\po p_1^{(j)},\ldots,\po p_{\ar{f}}^{(j)}$.
It may happen that two occurrences, say $f^{(j)}$ and $f^{(j')}$ have the same inputs
$(\po p_1^{(j)},\ldots,\po p_{\ar{f}}^{(j)})=(\po p_1^{(j')},\ldots,\po p_{\ar{f}}^{(j')})$.
We identify them by the equivalence relation $j\sim j'$.
Note here that the number
$\sum_{f \mbox{\tiny \ occuring in\ } \po p} \card{N_f}$
determines the size of the polynomial $\po p$,
while $\sum_{f \mbox{\tiny \ occuring in\ } \po p} \card{N_f}/\!\!\sim$ is (the lower bound for) the size of the corresponding algebraic circuit.
Now, going recursively down with the formula (\ref{rec:basic})
along the construction of the polynomial/circuit from the basic operations we end up with
\[
\h{\po p}(\o{d'}) =
\sum_{f} \sum_{j\in N_f} \beta^{(j)}_f \cdot \h{f}(\po{p}_1^{(j)}(\o{d'}),\ldots,\po{p}_{\ar{f}}^{(j)}(\o{d'})),
\]
where the $\beta^{(j)}_f$'s are algebraic combinations of appropriate coefficients of the form
$\alpha_i^g$'s for the basic operations $g$'s.
Since the values $\h{f}(\po{p}_1^{(j)}(\o{d'}),\ldots,\po{p}_{\ar{f}}^{(j)}(\o{d'}))$
are the same for the $j$'s from one $\sim$-class we can sum corresponding
$\beta^{(j)}_f$'s to shorten our expression appropriately
(this may be important for the size of the circuit we are going to construct).

\medskip
After this preparatory algebraic work we return to our construction of the boolean
$\ccand_d\circ\ccmod_{m}\circ\ccmod_p$-circuit
computing the boolean function of the program
$\progg{\po p}{n}{\iota}{\ccc}$.
Note first that for $\o b\in \bool^n$ we have
$\progb{\po p}{\iota}{\ccc}(\o b)=\true$   
iff $\prog{\po p}{\iota}(\o b)$ coincides with $\ccc$ on both coordinates: on $M$ and on $D'$,
or in other words iff $\progb{\po p/\beta}{\iota/\beta}{\ccc/\beta}(\o b)=\true$
and $\progb{p^M}{\iota}{\ccc^M}(\o b)=\true$,
for the quotient program $\progg{\po p/\beta}{n}{\iota/\beta}{\ccc/\beta}$
and the $p^M$-program $\progg{p^M}{n}{\iota}{\ccc^M}$ with the function
$p^M:D^k\map M$ determined by $p^M(d_1,\ldots,d_k)=(\po p(d_1,\ldots,d_k))^M$, respectively.
Thus the circuit computing $\progb{\po p}{\iota}{\ccc}$ will be obtained first by gluing  (by the binary $\ccand$-gate) the $\ccand_d\circ\ccmod_{m}\circ\ccmod_p$-circuits
for $\progb{\po p/\beta}{\iota/\beta}{\ccc/\beta}$ and $\progb{p^M}{\iota}{\ccc^M}$ and then eliminating this binary $\ccand$-gate with the help of \cref{lm:apply_func}.

The circuit computing the function $\progb{\po p/\beta}{\iota/\beta}{\ccc/\beta}$ is supplied by our recursion hypothesis for the program $\progg{\po p/\beta}{n}{\iota/\beta}{\ccc/\beta}$.
However the circuit for $\progb{p^M}{\iota}{\ccc^M}$ requires more work.
Our previous considerations combined with the instructions $\iota(x_i)=(\b^{x_i},a^{x_i}_0,a^{x_i}_1)$ allow us to compute the projection of
$\prog{p^M}{\iota}(\o b)$ onto $M$ to be
\[
\prog{p^M}{\iota}(\o b) =
\sum_{i=1}^{\ar{\po p}} \alpha^{\po p}_i\cdot (a^{x_i}_{\b^{x_i}})^M +
\\
\sum_{f} \sum_{j\in N_f} \beta^{(j)}_f \cdot
 \h{f}\left(\po{p}_1^{(j)}(\ldots,(a^{x_i}_{\b^{x_i}})^{D'},\ldots),\ldots,
\po{p}_{\ar{f}}^{(j)}(\ldots,(a^{x_i}_{\b^{x_i}})^{D'},\ldots)\right)
\]
or in other words
\[
\prog{p^M}{\iota}(\o b) =
\sum_{i=1}^{\ar{\po p}} \alpha^{\po p}_i\cdot (\prog{x_i}{\iota}(\o b))^M +
\\
\sum_{f} \sum_{j\in N_f} \beta^{(j)}_f \cdot
\h{f}\left(\left(\prog{\po p_1^{(j)}}{\iota}(\o b)\right)^{D'},\ldots,
\left(\prog{\po p_{\ar{f}}^{(j)}}{\iota}(\o b)\right)^{D'}\right),
\]
where $\prog{x_i}{\iota}$ is the inner function determined by applying instruction $\iota$ to the polynomial\break $\po q_i(x_1,\ldots,x_k)=x_i$.
By \cref{lm:and-sum}, the summands $\alpha^{\po p}_i\cdot (\prog{x_i}{\iota}(\o b))^M$,
as (essentially) the functions of the form $\bool\map M$, can be computed
by $\sumpk{p}{k}$-circuits of size $O(1)$,
so that also their sum can be computed
by $\ar{\po p}$-ary $\sumpk{p}{k}$-circuits of size $O(1)$.
To represent the other summands in the above display by
$\ccand_d\circ\ccmod_{m}\circ\ccmod_p
\circ\ccand_{\card{D'}}\circ\sumpk{p}{k}$-circuits
we first replace
\[
\h{f}\left(\left(\prog{\po p_1^{(j)}}{\iota}(\o b)\right)^{D'},\ldots,
\left(\prog{\po p_{\ar{f}}^{(j)}}{\iota}(\o b)\right)^{D'}\right)
\]
by

\begin{multline*}
\fcirc{\h{f}} \left( \progb{\po p_1^{(j)}}{\iota/\beta}{d_1/\beta}(\o b), \ldots, \progb{\po p_1^{(j)}}{\iota/\beta}{d_s/\beta}(\o b), \ldots, \right. \\
\quad \left. \progb{\po p_{\ar{f}}^{(j)}}{\iota/\beta}{d_1/\beta}(\o b), \ldots, \progb{\po p_{\ar{f}}^{(j)}}{\iota/\beta}{d_s/\beta}(\o b), \ldots \right)
\end{multline*}

\noindent where $d_1,\ldots,d_s$ is the transversal of the quotient $D/\beta$.
By induction hypothesis each of the functions
$\progb{\po p_i^{(j)}}{\iota/\beta}{d_l/\beta}:\bool^n\map D'$
is computable by $\ccand_d\circ\ccmod_{m}\circ\ccmod_p$-circuit
of size $O(\csize)$.
On the other hand, again by \cref{lm:and-sum},
the function $\fcirc{\h{f}}:\bool^{\card{D'}\cdot\ar{f}}\map M$ is computable by
$\ccand_{\card{D'}\cdot\ar{f}}\circ\sumpk{p}{k}$-circuit of size $O(2^{\maxar{D}\cdot s})$,
or simply $O(1)$ as it does not depend on $\csize$ or $\ell$.
Summing all those $\ell=\sum_{f \mbox{\tiny \ occuring in\ } \po p} \card{N_f}$
summands (with appropriate scalars) we end up
with $\ccand_d\circ\ccmod_{m}\circ\ccmod_p
\circ\ccand_{\card{D'}\cdot\ar{f}}\circ\sumpk{p}{k}$-circuit of size $O(\csize\cdot\ell)$
computing $\prog{p^M}{\iota}$.
Now \cref{lm:5to3} allows us to replace the 5-level circuit,
in which $\sumpk{p}{k}$ is replaced by $\sumpk{p}{k}^{c^M}$,
computing $\progb{p^M}{\iota}{c^M}$ by a
$\ccand_d\circ\ccmod_{m}\circ\ccmod_p$-circuit
of size $O((\csize\cdot\ell)^{\nu p^3\card{D'}\maxar{D}})$.
Finally \cref{lm:apply_func} allows us to glue, by binary $\ccand$,
this circuit with the $\ccand_d\circ\ccmod_{m}\circ\ccmod_p$-circuit
for $\progb{\po p/\beta}{\iota/\beta}{\ccc/\beta}$
(of size $O(\csize)$) to end up with a circuit
of size $O((\csize\cdot\ell)^{2\nu p^4\card{D'}\maxar{D}})$
computing $\progb{\po p}{\iota}{\ccc}$ in $\m D$.

Note that $\card{D}=\card{D'}\cdot p^\nu$ so that $\nu\cdot\card{D'}\leq\card{D}$
and consequently the size of the final circuit can be bounded by
$O((\csize\cdot\ell)^{2p^4\card{D}\maxar{D}})$,
as claimed in (\ref{section:easy}.\ref{mod-beta}).
\end{proof}

Note here that the degree of the polynomial bounding the size of the circuit,
i.e.  $O(\card{A}\cdot\maxar{A})^{O(\log\card{A})}$ have two sources.
The $\log\card{A}$ comes from the number of iterative use of claim (\ref{section:easy}.\ref{mod-beta}), i.e. from the hight of the congruence lattice of $\m A$.
The $(\card{A}\cdot\maxar{A})$-part is a consequence of arity of a binary expansion
$\fcirc{f} : \bool^{\card{A}\cdot\ar{f}} \map A$ coding basic operations $f$.
The careful reader can easily note some room for improvement here to
$O(\log\card{A}\cdot\maxar{A})^{O(\log\card{A})}$ by coding the elements of $A$ with $\log\card{A}$ bits and therefore shrinking the arity of $\fcirc{f}$ to $\log\card{A}\cdot\ar{f}$.
\section{Final Remarks}

Constant Degree Hypothesis plays crucial role in our proofs of \rptime algorithms existence. One can ask if this assumption is really needed. In fact, there are some  unconditional results. For example very recent paper \cite{KawalekKK19} shows that \ceqv{\m A} is in $\ptime$ whenever $\m A$  from CM is $2$-nilpotent, i.e. it has abelian congruence with abelian quotient. Also supernilpotent algebras admit (unconditional) polynomial-time algorithm for \csat{}/\ceqv{} \cite{aichmud-2010, komp2017, IdziakK22}, and if we allow random bits the time complexity drops down to linear \cite{KawalekK}. Unfortunately, it is not hard to construct for a given $d$, $m$, $p$ an algebra $\m A$ with supernilpotent rank equal $2$ such that functions computable by $\ccand_d \circ \ccmod_m \circ \ccmod_p$-circuit are exactly functions computable by, not too long, programs over $\m A$.  This, together with our results, shows that (under ETH) showing unconditional algorithms solving \progcsat{} for algebras from congruence modular variety with supernilpotent rank equal $2$ is equivalent to proving CDH. Hence, the natural question is if CDH holds.
\begin{prob}
Prove or disprove the Constant Degree Hyphothesis.
 \end{prob}

The natural next step in our investigations is to go outside congruence modular realm. The first problem in such a case is that Tame Congruence Theory and Commutator Theory (our heavily used tools) for arbitrary algebras do not work as well as for algebras from congruence modular varieties.  Moreover \cite[Example 2.8]{IdziakK22} shows that there is an algebra $\m A$ (not contained in congruence modular variety) and its congruence $\sigma$ such that $\csat{\m A}$ is in $\ptime$, while $\csat{\m A/\sigma}$ is \npc. This suggests that we cannot expect a nice characterization of polynomial-time cases for \csat{} outside CM. On the other hand, we are not aware of any such examples for \progcsat{} and \ceqv{}. In fact we saw that the hardness of $\progcsat{\m A/\sigma}$ implies the hardness for $\progcsat{\m A}$. This gives hope for characterization of tractable cases of \progcsat{} for general finite algebras.

\begin{prob}
 Characterize finite algebras with tractable \progcsat{}/\ceqv{}.
 \end{prob}

\bibliographystyle{alpha}
  \bibliography{equations}
\end{document}